\documentclass[10pt, twosides]{amsart}

\usepackage{defs, amsaddr}
\graphicspath{{images/}}

\title[Optimal multiplexing on matrix Lie groups]{On optimal multiplexing of an ensemble of discrete-time constrained control systems on matrix Lie groups}
\author[C. Maheshwari, S. Srikant, and D. Chatterjee]{Chinmay Maheshwari}
\address{EECS\\ UC Berkeley\\ California, USA}
\author{Sukumar Srikant and Debasish Chatterjee}
\address{Systems and Control Engineering\\ IIT Bombay, Powai\\ Mumbai 400076, India\\ \url{http://www.sc.iitb.ac.in/~srikant}\\ \url{http://www.sc.iitb.ac.in/~chatterjee}}

\email{chinmay\_maheshwari@berkeley.edu,\{srikant.sukumar,dchatter\}@iitb.ac.in}
\thanks{This work was partially supported by the grant 17ISROC001 from the Indian Space Research Organization.}

\begin{document}
	\begin{abstract}
		{We study a constrained optimal control problem for an ensemble of control systems. Each sub-system (or plant) evolves on a matrix Lie group, and must satisfy given state and control action constraints pointwise in time. In addition, certain multiplexing requirement is imposed: the controller must be shared between the plants in the sense that at any time instant the control signal may be sent to only one plant.  We provide first-order necessary conditions for optimality in the form of suitable Pontryagin maximum principle in this problem. Two numerical experiments are presented: first, for a system of two satellites; second, for a system of two underwater vehicles  performing energy optimal maneuvers under the preceding family of constraints.}
	\end{abstract} 

	\maketitle
	\section{Introduction}
	\label{sec: Introduction}

This article studies a problem of constrained optimal control of an ensemble of \dt \ control systems that evolve on a class of non-flat manifolds, namely, matrix Lie groups, that arise naturally in models of mechanical systems. These systems are assumed to be controlled over a network that permits the transmission of control signals to at most one plant at each time instant; we may view this stipulation as a constraint dictated by the network. In other words, there is a \emph{multiplexing} (or \emph{scheduling}) scheme that selects one plant from the ensemble at each time instant, and that particular plant is controlled at that instant while the rest of the plants evolve under zero control. Such multiplexers arise naturally in any situation where a central server must cater to a larger number of systems than the number of available processors. For instance, multiplexers (or schedulers) are present in every microprocessor that drives our computers, and they schedule jobs for each core according to priorities. From a control-theoretic perspective, if the control server is incapable of parallel processing but is assigned to control an ensemble of control systems, it must process the control tasks serially and, consequently, employ a multiplexing scheme to dispatch the control signals to different plants in the ensemble \cite{Kumar2018sparseLinear}. A safety-critical application of such a setup is found in medicine, where a fleet of micro-robots are injected into the blood stream of a patient for targeted drug delivery \cite{ref:ChoJinCap-15}; each micro-robot is too small to carry its own control and communication unit, and their control signals must be constructed in a way such that only one of them is controlled at any instant of time.

The importance of the co-design of control-multiplexing schemes can hardly be overstated in the context of networked control \cite{rehbinder2004scheduling}. Typically, a networked control system consists of a collection of sub-systems (or plants) with actuators and sensors, all connected over a shared communication channel; such systems arise in a variety of applications including automobiles, aircraft, spacecraft, the manufacturing and process industry, etc. Different approaches to the co-design of control-multiplexing schemes have been proposed in the literature, and these approaches can be broadly classified based on whether the multiplexing is (a) periodic \cite{gorges2007optimal} or (b) aperiodic \cite{reimann2012novel, LQQOptimalMultiplex, Kumar2018sparseLinear}, and a large body of literature is available today that pertains to both of these types. For instance, the control and scheduling co-design problem was formulated as a Lyapunov-based stabilization problem for switched linear systems in \cite{reimann2012novel} and then transformed into an optimization problem with linear matrix inequality (LMI) constraints; optimal control and scheduling of NCSs that are modeled as \dt \ switched linear systems have been presented in \cite{LQQOptimalMultiplex}, where the authors minimize a quadratic performance criterion via a receding horizon scheme and scheduling strategy and the resulting problem is solved via dynamic programming; recently in \cite{Kumar2018sparseLinear} the authors have implemented sparse optimal scheduling for continuous-time linear systems. In the article at hand we take a step beyond by \emph{not} stipulating the multiplexing algorithm to be of either of the two types (a) and (b); instead, we insist that the multiplexing algorithm is optimal.

Constraints on the states and the control actions are omnipresent in realistic control systems. Consider, for instance, a satellite in outer space that is commanded to undergo a change of orientation to align its sensors to a particular star. During such a maneuver the control actions at one's disposal are limited by the mechanical capabilities of the actuators in the satellite, and in order to ensure the safety of the mechanical components on board (that may fail if the angular velocities exceed a safe limit), the momenta of the satellite during such maneuvers must not exceed given safety thresholds. Here we have control action constraints due to physical limitations and state constraints that must be satisfied for safety. Designing controllers that execute given tasks while satisfying such state and control constraints is a non-trivial and challenging task. (Indeed, most of the available literature on control-multiplexing co-design problems consider the underlying system dynamics to be linear with the notable exception of \cite{hashimoto2017collision}, where the authors propose a scheduling algorithm for model predictive control (MPC) of an ensemble of nonlinear continuous time systems with the constraints on the control actions; constraints on the states of the system have not, however, been considered there.) The problem gets further complicated when some form of optimality is demanded over either the state trajectories, the control action trajectories, or both.

The Pontryagin maximum principle (PMP) is a widely used tool that provides the first order necessary conditions for the optimality of control systems, which takes the form of a set of nonlinear equations that may be solved numerically to obtain the candidate optimal control trajectories. The accuracy of such numerical schemes depends largely on the discretization of the underlying dynamics of the systems. For systems that evolve on non-flat configuration spaces, such discretization procedures are non-trivial, and \dt \ models should preferably be derived using the ideas of discrete mechanics \cite{MR2009697}. Moreover, since control algorithms are applied digitally today, it is highly desirable to directly work with \dt \ control strategies, especially those that ensure a high level of accuracy and fidelity such as discrete mechanics. A series of studies centered around various \dt \ versions of the PMP have recently been conducted with the intention of employing such PMPs as general platforms for constrained state-action trajectory synthesis. \cite{phogat2016discrete} addressed optimal control of \dt \ systems evolving on matrix Lie groups under state-action constraints, and this work has been extended to cater to more general systems evolving on smooth manifolds in \cite{Assif2018}. Frequency constraints on the control action trajectories were included in the list of constraints in \cite{paruchuri2017discrete} for systems on Euclidean spaces, and then extended in \cite{KPPCP} to systems evolving on matrix Lie groups. All these results derive essentially from the by-now classical work \cite{bolt1975method}.

In this article we continue this line of study by expanding the scope of the preceding results by incorporating a new type of constraint. Here we provide the first order necessary conditions for optimally controlling an ensemble of \dt \ control systems evolving on matrix Lie groups while satisfying prescribed state-action constraints and controlled remotely via a single shared control channel. Optimal control problems of such types find standard applications, e.g., in the control of quadcopter fleets \cite{ritz2012cooperative, tang2017multi}, of groups of satellites \cite{scharf2003survey}, etc., and nonstandard applications, e.g., in the control of medical micro-robots \cite{ref:ChoJinCap-15}, where providing the control input simultaneously to every subsystem is neither feasible nor desirable. 

Our specific contributions are summarized below: 
\begin{itemize}[leftmargin=*]
	\item We provide a PMP for constrained optimal control of an ensemble of \dt \ control systems that evolve on matrix Lie groups, where
	\begin{itemize}
		\item constraints on the states at each instant of time are present,
		\item constraints on the control actions at each instant of time are present, and
		\item multiplexing constraints imposed by shared computational/communication resources that are used to command our ensemble of control systems are present, taking the form that only one system is controlled at any given time while the others evolve under zero control.
	\end{itemize}
	\item Our results are designed to work with \dt \ dynamics derived via \emph{discrete mechanics} \cite{MR2009697}, thereby preserving the underlying manifold structure as well as certain important system invariants. This important and desirable feature contributes to greater accuracy and fidelity than otherwise for \dt \ implementation.
\end{itemize}
To the best of our knowledge there is no prior work on the class of constrained problems treated here although it is quite a natural setting. While our results may be employed by numerical algorithms to solve for optimal multiplexed control trajectories via the so-called indirect method \cite{Trelat2012} (an endeavour to be pursued separately), they may also be used to verify the optimality of solutions obtained via third party ``black-box'' solvers.

Our article exposes as follows: In \S \ref{sec: Problem Formulation} we formulate the problem statement and \S \ref{sec: Main Result} contains our main result, a proof of which is given in \S \ref{sec: Proof}. \S \ref{sec: Num Exp} contains a detailed numerical experiment to illustrate the efficacy of our technique. We employed the freely available NLP solver \textsf{CasADi} for two separate cases: for a pair of satellites and a pair of underwater vehicle, sharing a single control channel and performing energy optimal point-to-point ballistic reachability maneuvers under constraints on the states and on the magnitudes of the control actions. The outputs of the solver were verified to be optimal by employing our main result.

	\textbf{Notations}: For a positive integer \(\nu\), the transpose of a vector \(\xvec[] \in \R^\nu\) is denoted by \(\xvec[]\top\), \(\norm{\xvec[]}\) denotes standard Euclidean norm defined by \(\norm{\xvec[]}\) \defas \(\sqrt{\xvec[]\top \xvec[]}\). For any positive integer \(\intK\) we define \([\intK] \defas \{0,1,2,\dots,\intK\}\) and \({[{\intK}]\dual}  \defas \{1,2,\dots,\intK\}\) . For two positive integers \(M,m \) such that \(M \ge m\) the number of ways of choosing \(m\) distinct elements from a set containing \(M\) distinct elements is written as \({M \choose m}\). We denote the cardinality of a finite set \(S\) by \(\card{S}\). For two positive integers \(\nDim[1]\) and \(\nDim[2]\), we define a matrix \(\linTF \in \R^{\nDim[1]\nDim[2] \times \nDim[1]}\) by
	\begin{equation*}
		\linTF_{i,j} \Let \begin{cases}
		1  \quad \text{if } (j-1)\nDim[2] +1 \le i \le j \nDim[2],\\
		0  \quad \text{otherwise}.
		\end{cases}
	\end{equation*}
	Let \(\vect[1] \in \R^{\nDim[1]}\), \(\vect[2] \in \R^{\nDim[2]}\). We define
	\begin{align*}
		\vect[1] \opr \vect[2] \defas \Big \langle \linTF \vect[1], (\underbrace{\vect[2]\top,\vect[2]\top\dots, \vect[2]\top}_{\nDim[1] \text{times}})\top \Big\rangle \in \R.
	\end{align*}
	We denote the direct sum of two vector spaces \(\VS[1]\) and \(\VS[2]\) by \(\VS[1] \DirSum \VS[2]\). For two vectors \(\xvec[1], \xvec[2] \in \R^m\), where \(m\) is some positive integer, we define the Hadamard product\footnote{Consider two vectors \(\xvec[1], \xvec[2] \in \R^n \). The Hadamard product (a.k.a.\ the Schur product) of \(\xvec[1]\) and \(\xvec[2]\) produces the vector \(v \in \R^n\) given by the entry-wise multiplication of \(\xvec[1]\) and \(\xvec[2]\), and the Hadamard product for matrices is defined similarly.} between them by \(\xvec[1] \HadPr \xvec[2]\). For a vector \(\xvec[] \in \R^n\), where \(n\) is a positive integer, we use \(\vecCompare{{\xvec[]}}{0}\) to denote that all the components of \(\xvec[]\) are non-positive.
	
	\section{Problem Formulation}
	\label{sec: Problem Formulation}
We start with the description of a \emph{multiplexed control system}. It is a dynamical system comprising of \(\numPlants\) plants controlled by a server that can transmit a control signal at a given time to at most one among \(\numPlants\) plants via a transmission channel. The situation is as shown in Figure \ref{fig: ServerChannel}. {We assume that the transmission of control signals from the server to the plants is without any delay. Further, we assume that there is no uncertainty in the system. To motivate the discussion, consider an ensemble of satellites that are required to orient some sensors to a distant object in the space. The discrete-time attitude dynamics of each satellite is given by (\cite{phogat2017discrete})
\begin{equation} \label{eq: IntroSatellite}
\begin{aligned}
	\SCRotm[t+1] &= \SCRotm[t]\SCGrpDyn[t](\SCRotm[t],\SCAngM[t]), \\
	\SCAngM[t+1] &=  \SCAngM[t] + \SCstep\SCCont[t],
\end{aligned}
\end{equation}
where \(t\) is an integer, \(\SCstep\) is the discretization step size, \(\SCRotm[t] \Let \SCRotm[](th) \in \R^{3\times3}\) denotes the rotation matrix that encapsulates the orientation information of the satellite at time \(th\), \(\SCAngM \Let \SCAngM[](th)\in \R^3\) represents the angular momentum of the satellite at time \(th\), \(\SCCont \in \R^3 \) is the control input to the system in the form of torque at time \(th\), \(\SCGrpDyn\) is the map depicting the dynamics of the system on the Lie group (which is SO\((3)\) in this case) at time \(th\). For the sake of brevity, this point onwards we shall omit the step size \(\SCstep \) while referring the time instant \(t\SCstep \) and call it simply the time instant \(t\).
}

Consider an ensemble of control systems comprising of \(\numPlants\) plants. For \(\iRange \), fix positive integers \(\ithSysDimState\), \(\ithSysDimControl \) and \(\ithSysDimConstraint \). Motivated by \eqref{eq: IntroSatellite}, we assume more generally that the dynamics of the \(\ith\) plant in the ensemble are split into two parts, one that evolves on a matrix Lie group \(\grpMfdIndex\), while the other evolves on a Euclidean space \(\eucldStateSpace\) \cite{phogat2016discrete}. More precisely, the \dt\ control system corresponding to the \(\ith\) plant evolves on a configuration space \(\ithConfMfd\), and is described by the recursion
\begin{equation}
	\label{eq:ithSysDynamics}
	\begin{cases}
		\grpIndexT[t+1] = \grpIndexT \,\jumpMapIndexT\bigl(\grpIndexT, \stateIndexT\bigr), \\
		\stateIndexT[t+1] = \ithSysDyn\bigl(\grpIndexT, \stateIndexT, \contIndexT\bigr),
	\end{cases}
	\qquad \timeRange,\ \iRange,
\end{equation}
with the following data:
\begin{enumerate}[label=(\alph*), leftmargin=*, widest=b, align=left]
	\item \(\grpIndexT\) and \(\stateIndexT\) are the state variables corresponding to the \(\ith\) plant at the time instant \(t\), with \(\grpIndexT\) residing on the matrix Lie group \(\grpMfdIndex\) and \(\stateIndexT\) residing on the Euclidean space \(\eucldStateSpace\), respectively;
	
	\item \(\contIndexT \ \belongsTo \ \ithAdmControl \ \subsetTo \ \eucldControlSpace\) is the control action injected to the \(\ith\) plant at time \(t\), where \(\ithAdmControl\) is a convex and compact set (containing \(0 \belongsTo \eucldControlSpace\)) of admissible control actions that may be applied to the \(\ith\) plant at any time instant.\footnote{Recall that a set \(S \subsetTo \R^n\) is \textit{convex} if for any two point \(\xvec[], \yvec[] \belongsTo S\), we have \((1-\lam) \xvec[] + \lam \yvec[] \belongsTo S\) for every \(\lam \in [0,1]\). A non-convex set is one that is not convex. A set \(S \subsetTo \R^n\) is \emph{compact} if and only if it is closed and bounded; this is a consequence of the Heine-Borel Theorem (\cite[Theorem 6.1.1.]{RealAnalysis}).}
	
	\item \(\jumpMapIndexT: \ithConfMfd \ra \grpMfdIndex\) is a smooth map describing the part of the dynamics of the \(\ith\) plant on the matrix Lie group \(\grpMfdIndex\);
	
	\item \(\ithSysDyn: \ithDynMfd \ra \eucldStateSpace\) is a smooth map governing the part of the dynamics of the \(\ith\) plant on the Euclidean space \(\eucldStateSpace\);
\end{enumerate}

The task of the multiplexer is to select from the aforementioned ensemble \eqref{eq:ithSysDynamics} of systems, at each time instant \(t\), the index \(i\) of the plant to which the control action at time \(t\) has to be applied. 

\begin{figure}
	\centering
	\includegraphics[width=0.8\linewidth]{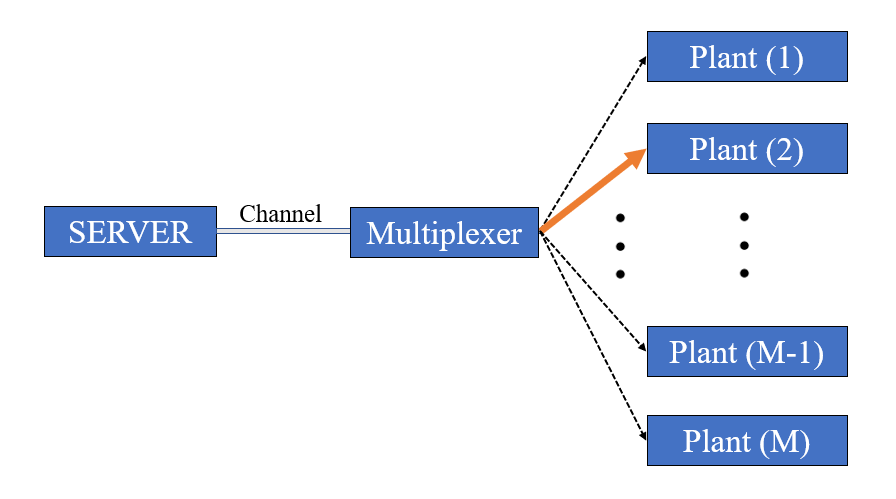}
	\caption{Schematic of a multiplexed control system with a server that sends control signal to the multiplexer, which in turn devolves the control signal to appropriate plant, to minimize a predefined cost incurred in the process. In the above figure the multiplexer chooses Plant(2) out of \numPlants \ plants.}
	\label{fig: ServerChannel}
\end{figure}
We regard the ensemble of systems \eqref{eq:ithSysDynamics} as a joint system in a natural way: we define \(\sysDimControl \defas \nuExpression\) and \(\sysDimState \defas  \nxExpression\) to be the dimensions of the admissible joint control action set and the Euclidean space for the joint system of \(\numPlants\) plants. The admissible joint set of control actions for the joint control system is
\begin{equation}
	\label{eq: AdmissibleActionSet}
	\uhash \defas \bigcup_{k=1}^{\numPlants} \bigg( \zeroI[1] \times \zeroI[2] \times \dots \ithAdmControl[k] \dots \times \zeroI[\numPlants] \bigg),
\end{equation}
where \(\zeroI[i] \belongsTo \eucldControlSpace\). The set in \eqref{eq: AdmissibleActionSet} is a ``star''-shaped admissible control action set due to which the multiplexing constraint is satisfied (that is, at any instant only one out of \(\numPlants \) plants is provided with the control action). Furthermore, the joint control system evolves on the matrix Lie group \({\gTot \defas  \qCollection}\) and the Euclidean space \(\xTot \defas \RPower[\sysDimState]\). 
\begin{remark}
	We equip the cartesian product \(\gTot\) with the direct product group structure. Recall that for two groups \(\grpMfdIndex[1], \grpMfdIndex[2]\) with \(\xvec[1], \yvec[1] \belongsTo \grpMfdIndex[1] \)and \(\xvec[2], \yvec[2] \belongsTo \grpMfdIndex[2]\) we define the group operation on \(\gTot \defas \grpMfdIndex[1] \times \grpMfdIndex[2]\) to be \((\xvec[1], \xvec[2])(\yvec[1], \yvec[2]) = (\xvec[1]\yvec[1], \xvec[2]\yvec[2])\). The identity element in \(\gTot\) is denoted by \(\IdEle[]\) where \(\IdEle[] \defas ( \IdEle[1], \IdEle[2])\). Here  \(\IdEle[1]\) and \(\IdEle[2]\) are the identity elements of \(\grpMfdIndex[1]\) and \(\grpMfdIndex[2]\), respectively. For Lie groups the product and inverse operations are smooth by definition, and a finite product of matrix Lie groups is also a matrix Lie group \cite[Example 5.1.3]{ref:RudSch-13}. 
\end{remark}   

Against the preceding backdrop, we formulate our optimal control problem:
\begin{equation}
\label{eq:ProbFormulation}
\begin{aligned}
	& \minimize_{(\uT)_{t=0}^{\nS-1}} && \mathcal{C}(\qHat,\xHat,\uHat) \Let \sum_{t=0}^{\nS-1} \ithCost\big(\qT, \xT, \uT\big) + \ithCost[\nS]\big(\qT[\nS], \xT[\nS]\big) \\ 
	& \text{subject to} &&
\begin{cases}
	\text{dynamics \eqref{eq:ithSysDynamics} } & \text{for all} \ \timeRange \ \text{and} \ \iRange, \\
	\uT \belongsTo \uhash & \text{for all} \ \timeRange,\\
	\stateConst \big(\grpIndexT, \stateIndexT\big) \leq 0 & \text{for all}  \ \timeRangeFullSt \ \text{and} \ \iRange, \\
	\left(\qT[0],\xT[0]\right)=\big(\ob{\qT[0]},\ob{\xT[0]}\big),
\end{cases} 
\end{aligned}
\end{equation}
with the following data:
\begin{enumerate}[label=(\alph*), leftmargin=*, widest=b, align=left]
	\item For every \(t \in [\nS]\), \(\qT[t] \in \gTot\) and \(\xT[t] \in \xTot \). Similarly, for \(t \in [\nS-1]\) we have \(\uT[t] \in \uhash\);
	
	\item \(\qHat \defas (\qT[0], \qT[1], \dots , \qT[\nS])\),  \(\xHat \defas ( \xT[0], \xT[1], \dots , \xT[\nS] )\), and \(\uHat \defas  ( \uT[0], \uT[1], \dots , \uT[\nS-1] )\) are ordered tuples of elements in \(\gTot\), \(\xTot\) and \(\uhash\), respectively;
	
	\item \(\ithCost: \DynMfd \ra \R\) denotes the cost incurred at each time instant \(\timeRange\);
	
	\item \(\ithCost[\nS]: \ConfMfd \ra \R\) denotes the cost incurred at the final time instant \(t = \nS\);
	
	\item \(\stateConst: \ithConfMfd \ra \eucldConstraintSpace\) denotes state constraints that needs to be satisfied by the \(\ith\) plant for every \(\timeRangeFullSt\);
	
	 \item \((\ob{\qT[0]},\ob{\xT[0]})\) denote the user-defined initial conditions. 
\end{enumerate}

	\section{Main Results}
	\label{sec: Main Result}
	\subsection{Preliminaries}\label{subsec:firstMR}
We begin with several definitions that will be needed to state our main results. We adhere throughout this article to the definitions of smooth manifolds, tangent spaces and cotangent spaces given in \cite[Chapter 1]{ref:RudSch-13}. For any point \(q \) on a smooth manifold \(\mfd \), we denote the tangent space and cotangent space at the point \(q\) by \(\tanLift{q}{\mfd} \) and \(\cotanLift{q}{\mfd} \)  respectively.

\begin{definition}[{\cite[p.\ 124]{marsden1995introduction}}]
	Let \(\fmap : \mfd \ra \R\) be a smooth function defined on a smooth manifold \(\mfd\). The \emph{directional (Lie) derivative} of \(\fmap\) at a point \(q \in \mfd\) along a vector \(\tanVec \in \tanLift{\basePt}{\mfd}\) is the map
	\[
		\tanLift{\basePt}{\mfd} \ni \tanVec \mapsto \mathcal{D} \fmap(\basePt) \tanVec \defas \frac{d}{dt} \Big|_{t=0}  \fmap \left(\curveMfd\left(t\right)\right) \in \R,
	\]
	where \(\R \ni t \mapsto \curveMfd(t) \in \mfd\) is any smooth map satisfying \(\curveMfd(0) = \basePt\) and \(\left.\frac{d}{dt} \right|_{t=0} \curveMfd(t) = \tanVec\).
\end{definition}

\begin{definition}[{{\cite[p.\ 173]{marsden1995introduction}}}]
	Let \(\leftAction: \grpG \times \grpG \ra \grpG\) be the left action defined on a Lie group \(\grpG\), i.e., \(\grpG \ni \gp \mapsto \leftAction_\grp(\gp)\defas \grp \gp \in \grpG\)  for any \(\grp \belongsTo \grpG\). The \textit{tangent lift} \(\tanLift{}{\leftAction}:\grpG \times \tanLift{}{\grpG} \ra \tanLift{}{\grpG}\) of \(\leftAction\) is the action defined, for \(\gOne \belongsTo \grpG \) and \(\TgG \belongsTo \tanLift{\gOne}{\grpG}\), by 
	\[
		\left(\grp,\left(\gOne,\TgG \right)\right) \mapsto \tanLift{}{\leftAction_{\grp}}\left(\gOne,\TgG \right) = \left(\leftAction_{\grp}(\gOne),\tanLift{\gOne}{\leftAction_{\grp}}(\TgG) \right).
	\]
	The \textit{cotangent lift} \(\cotanLift{}{\leftAction}: \grpG \times \cotanLift{}{\grpG} \ra \cotanLift{}{\grpG}\) of \(\leftAction\) is, similarly, the action defined, for \(\gOne \belongsTo \grp\) and \(\dTgG \belongsTo \cotanLift{\gOne}{\grpG}\), by
	\[
		\left(\grp,\left(\gOne,\dTgG \right)\right) \mapsto \cotanLift{}{\leftAction_{\grp}}\left(\gOne,\dTgG \right) = \left(\leftAction_{\grp}(\gOne),\cotanLift{\leftAction_{\grp}(\gOne)}{\leftAction_{\grp^{-1}}} (\dTgG) \right).
	\]
\end{definition}

\begin{definition}[{\cite[p.\ 245]{marsden1995introduction}}]
	Let \(\leftAction: \grpG \times  \grpG \ra  \grpG\) be the left action defined on a Lie group \(\grpG\) with the Lie algebra \(\lieAlg\). A vector field \(\grpG \ni \grpEle \mapsto \vf[](\grpEle) \in \tanLift{\grpEle}{\grpG}\) on \(\grpG\) is called \textit{left invariant} if for any \(\grp, \grpOne \in \grpG\) we have 
	\begin{align*}
	(\tanLift{\grpOne}{\leftAction_\grp})\vf[](\grpOne) = \vf[](\grp \grpOne).
	\end{align*}
	For each ${\XI \in \lieAlg}$, we define the canonical left invariant vector field ${X_{\XI}}$ on ${\grpG}$ by ${\vf[{\XI}](\grp) \Let \tanLift{\IdEle[]}{\leftAction_\grp(\XI)}}$ where $\IdEle[]$ is the identity element of the group \(\grpG\). 
\end{definition}

If \(\vf[{\XI}]\) is the canonical left invariant vector field corresponding to \(\XI \belongsTo \lieAlg\), then there is an unique integral curve \(\intCurve_{\XI}: \R \ra \grpG\) of \(\vf[{\XI}]\) starting at \(\IdEle[]\) such that \(\intCurve_\XI(0) = \IdEle[]\) and \(\frac{d}{dt}\intCurve_{\XI}(t) = \vf[{\XI}](\intCurve_{\XI}(t))\).

\begin{definition}[{\cite[p.\ 248]{marsden1995introduction}}]
	The \textit{exponential map} \(\exp : \lieAlg \ra \grpG\) on the Lie algebra \(\lieAlg\) (of the Lie group \(\grpG\)) \ is defined by
	\begin{align*}
		\exp(\XI) = \intCurve_{\XI}(1).
	\end{align*} 
	In case of matrix Lie groups, this object is the standard matrix exponential.
\end{definition}

\begin{definition}[{{\cite[p.\ 311]{marsden1995introduction}}}]
	\label{def: AdjointDef}
	The \textit{adjoint action} on the Lie algebra \(\lieAlg\) (of the Lie group \(\grpG\)) is the map 
	\[
		\lieg \times \lieAlg \ni \left(\grp,\lieAlgEle \right) \mapsto \ad{\grp}\lieAlgEle \defas \left.\frac{d}{ds}\right|_{s=0} \grp \exp(s \lieAlgEle)\grp^{-1}  \belongsTo \lieAlg.
	\]
		The \textit{coadjoint action }of \(\lieg\) on the dual of the Lie algebra, \(\dLieAlg\), is the dual of the adjoint action
	\[
		\lieg \times \dLieAlg \ni \left(\grp,\dLieAlgEle \right) \mapsto \coAd{\grp^{-1}} \dLieAlgEle \in  \dLieAlg,
	\]
	defined by
	\[
		\big\langle{\coAd{\grp^{-1}}\dLieAlgEle},{\lieAlgEle} \big\rangle = \big\langle{\dLieAlgEle},{\ad{\grp^{-1}} \lieAlgEle} \big\rangle \quad \text{for all} \ \lieAlgEle \in \lieAlg \text{ and } \dLieAlgEle \in \dLieAlg.
	\]
\end{definition}

\begin{definition}[{\cite[p.\ 29]{bolt1975method}}]
	The \emph{support cone} \(\suppCone{\convSet}{\apex}\) of a convex set \(\convSet \subsetTo \RPower[m]\) with apex at \(\apex \in \convSet\) is defined to be 
	\begin{align*}
		\suppCone{\convSet}{\apex} \defas \cls\bigg( \bigcup_{\alpha>0} \bigl\{\apex + \alpha(\coneEle-\apex) \big| \coneEle \belongsTo \convSet \bigr\} \bigg).
	\end{align*} 
\end{definition}

The following constitute key assumptions for our results:
\begin{assumption}[cf.\ \cite{phogat2016discrete}]
	\label{ass: Existence}
	\mbox{}
	\begin{enumerate}[label=(A-\roman*), leftmargin=*, widest=b, align=left]
		\item \label{asm:1} For every \(\iRange \), the maps \(\jumpMapIndexT\), \(\ithSysDyn\), \(\stateConst\), \(\ithCost\), and \(\ithCost[\nS]\) defined in \secref{sec: Problem Formulation} are smooth for all \(\timeRangeFull\).
		
		\item \label{asm:2} For every \(\iRange\), there exists  an open set \(\OpenSet[i] \subsetTo \ \ithLieAlg\), where  \(\ithLieAlg\) is the Lie algebra corresponding to the Lie group \(\grpMfdIndex\), such that
		
		\begin{enumerate}[label = (\alph*)]
			\item  the exponential map of the \(\ith\) plant \(\ithExpMap: \OpenSet[i] \ra \ithExpMap(\OpenSet[i]) \subsetTo \ \grpMfdIndex \) is a smooth diffeomorphism, and
			
			\item the integration step \(\jumpMapIndexT \ \belongsTo \ \ithExpMap(\OpenSet[i])\) for all ${\timeRangeFull}$. 
		\end{enumerate}
	
	\item \label{asm:3} For every $\iRange$, the set \(\ithAdmControl\) is nonempty, convex, and compact.
	\end{enumerate} 
\end{assumption}

\begin{remark}
	A few words about Assumption \ref{ass: Existence} are in order. Our main result utilizes the \dt{} PMP conditions developed in \cite{phogat2016discrete} at its core. In \cite{phogat2016discrete} the \dt{} PMP conditions were obtained by working with a local conical approximation of the feasible region in a small neighborhood of an optimal point via Boltyanskii's method of tents \cite{bolt1975method}. \ref{asm:1} ensures the existence of such local conical approximation of the feasible set. \ref{asm:2} gives the local representation of admissible trajectories in the Lie algebra, which is a vector space. \ref{asm:3} leads to a pointwise non-positivity condition on the gradient of the Hamiltonian (defined in \eqref{eq: Hamiltonian Form} below) over the set of feasible control actions. This set of assumptions are mild and standard in the literature.
\end{remark}

Before heading towards Theorem \ref{th: 1}, we define the following set
\begin{equation}
	\label{eqn: ustar2}
	\begin{aligned}
		\ustar \defas \ithAdmControl[1] \times \ithAdmControl[2] \times \dots \times \ithAdmControl[\numPlants] ,
	\end{aligned}
\end{equation}
and the map
\begin{align}
	\label{eq: constraintEq}
	\ustar \ni \uT[] \mapsto \auxDyn \defas \sum_{i = 1}^{\numPlants-1} \sum_{j = i+1}^{\numPlants} \bigg( \bigl\|\uSig[i]\bigr\|^2 \bigl\|\uSig[j]\bigr\|^2 \begin{pmatrix} 1 \\ 1 \end{pmatrix} + \uSig[i] \opr \uSig[j] \begin{pmatrix} 1 \\ -1 \end{pmatrix}   \bigg) \in \RPower[2],
\end{align}
where \(\uSig[i] \in \ithAdmControl[i]\) for all \(\iRange \). We relegate some important properties of the map \(z\) to Appendix \ref{app: app_z} that will be utilized in \S\ref{sec: Proof}.  

Recall from \S \ref{sec: Problem Formulation} that for each \(\iRange \) the \(\ith \) plant evolves on the configuration space \(\ithConfMfd\), and the joint control system evolves on \(\gTot \times \RPower[\sysDimState] \). We shall denote the Lie algebra corresponding to the matrix Lie group \(\grpMfdIndex\) by \(\ithLieAlg\) and the corresponding dual Lie algebra by \(\dIthLieAlg \). Likewise, we shall denote the Lie algebra corresponding to the joint matrix Lie group \(\gTot\) by \(\liea\) and the corresponding dual Lie algebra by \(\dliea\). In addition the following functions will be employed in Theorem \ref{th: 1} below:
\begin{enumerate}[label = (\roman*) ,leftmargin=*, align = left, widest = iii]
	\item \(\kappaMap : \gTot \ra \grpMfdIndex\) is the projection map that gives the element corresponding to the \(\ith\) group from the product matrix Lie group \(\gTot\). As discussed above, \(\gTot\) has the structure of product Lie group.
			
	\item \(\ithProjMapPi : \ithProjMapPiDef\)  is a projection map from the Lie algebra \(\ithProjMapPi\) of the joint matrix Lie group \(\gTot\) to the Lie algebra of the matrix Lie group \(\grpMfdIndex\). In fact, this is the tangent map \(\tanLift{\IdEle[]}{\kappaMap}\) associated with \(\kappaMap\) at the identity element of \(\gTot\), and the Lie algebra of product Lie groups is the direct sum of Lie algebras of the individual Lie groups \cite[Chapter 5]{ref:RudSch-13}. Therefore, \(\ithProjMapPi\)  is well-defined.

	\item \(\lambdaMap : \dliea \ra \dIthLieAlg\) is a map from the dual of the Lie algebra of \(\gTot\) to the dual of the Lie algebra of \(\grpMfdIndex\). Well-posedness of this map is immediate as the dual of the Lie algebra of a matrix Lie group is a vector space. From \cite[\S 20]{halmos2017finite} we know that the dual of a direct sum of vector spaces is isomorphic to the direct sum of the individual dual vector spaces. Therefore, \(\lambdaMap\) is well-defined.
	
	\item The map \([\nS-1] \ni t \mapsto \multiplexFunc(t) \in [\numPlants]\dual\) is the multiplexer function that decides the branch of the ``star''-shaped admissible joint control action set where the control action resides at each time instant.
	
	\item \(\PhiI : \grpMfdIndex \times \grpMfdIndex \ra \grpMfdIndex\) is the left action on the Lie group \(\grpMfdIndex\).
\end{enumerate}

\subsection{Main Result}

The following is our main result:
\begin{theorem}\label{th: 1}
	Let \(\uOptSeq\) be an optimal control sequence that solves \eqref{eq:ProbFormulation} and let \(\stateTraj\) be the corresponding state trajectory. For \(\NU \in \{-1,0\}\) and \(\auxAdj \in \RPower[2]\), we define the Hamiltonian 
	\begin{equation}
	\begin{aligned}\label{eq: Hamiltonian Form}
		& \hamdef \ni  \hamvar \mapsto  \\ 
		& \quad \ham{\left(\tau,\dualGrp,\dualEucld,\qT[],\xT[], \uT[]\right)} \coloneqq \NU\ithCost[\tau](\qT[],\xT[],\uT[]) + \sum_{i=1}^{\numPlants} \bigg( \Big \langle \iththeta, \ithExpMap^{-1}\big(\jumpMapIndexT[\tau](\grpIndexT[],\stateIndexT[])\big)\Big \rangle_{\ithLieAlg[i]}  \\
		& \quad\quad + \Big\langle \ithxi,\ithSysDyn[\tau](\grpIndexT[],\stateIndexT[],\contIndexT[]) \Big\rangle  \bigg) + \big \langle \auxAdj,  \auxDyn[{\uT[]}] \big \rangle \in \R,
	\end{aligned}
	\end{equation}
	where \(\dualGrp \Let \big(\iththeta[i]\big)_{i=1}^{\numPlants}\), \(\dualEucld \Let \big(\ithxi[i]\big)_{i=1}^{\numPlants}\), \(\qT[] \Let \big(\grpIndexT[]\big)_{i=1}^{\numPlants}\) , \(\xT[] \Let \big(\stateIndexT[]\big)_{i=1}^{\numPlants}\), and \(\ithxi \belongsTo {\big(\eucldStateSpace\big)}\dual\) and \(\iththeta  \belongsTo \ithLieAlg\dual\). For \(\timeRange\) we define the transformation  
	\begin{align*}
		\ithLieAlg\dual \ni \iththeta[i]_t \mapsto \RhoTI \Let \Big(\mathcal{D}\ithExpMap^{-1}\big((\varQ{t})^{-1}\varQ{t+1}\big) \circ \tanLift{\ithId}{\PhiI_{(\varQ{t})^{-1}\varQ{t+1}}}\Big)\dual\big(\iththeta[i]_t\big) \in \ithLieAlg\dual \quad \text{for }\iRange,
	\end{align*}
	and denote the extremal lift of the state-action trajectory \((\opQT, \opXT, \opUT)\) under the optimal control \(\opT{\uT[]}\) at each time instant \(t\) 
	\begin{align*}
	\optimalHamState \defas (t, \dualGrp, \dualEucld, \opQT, \opXT, \opUT).
	\end{align*}
	Then there exist an adjoint trajectory \(\big(\grpTheta, \grpXi\big)_{t=0}^{\nS-1} \subsetTo \grpConfMfdDual \) and covectors \(\bigl(\ithMu\bigr)_{t=1}^{\nS} \ \subsetTo \ \big(\eucldConstraintSpace\big)\dual\) for \(\iRange \), such that the following conditions hold: 
	\begin{enumerate}[leftmargin=*, label={\rm (MP-\roman*)}, widest=b, align=left]
		\item \label{Non-triviality cond} non-triviality: the adjoint variables \((\grpTheta, \grpXi)\) for all \(\timeRange \), the 
		covectors \(\ithMu\) \(( \)for all \(\timeRangeFullSt\) and \(\iRange\)\()\), the scalar \(\NU\), and the vector \(\auxAdj\) do not vanish simultaneously;
		\item \label{stateNDadjointDynamics} state and adjoint system dynamics for all \iRange:
 			\begin{align*}
				\text{states} & \begin{cases}
				\varQ{t+1}  = \varQ{t} \ithExpMap\Big(\ithProjMapPi\big(\mathcal{D}_{\dualGrp}\ham(\optimalHamState)\big)\Big) \\
				\opXT[t+1] = \mathcal{D}_{\dualEucld}\ham(\optimalHamState)
				\end{cases}  \\
				\text{adjoints} & \begin{cases}
				\RhoTI[t-1] = \coAd{\ithExpMap\big(-\ithProjMapPi(\mathcal{D}_{\dualGrp}\ham(\optimalHamState))\big)}{\RhoTI} + \\ \qquad \cotanLift{ \ithId}{\PhiI_{{\kappaMap(\opT{Q})}}}\bigg(\lambdaMap\Big(\mathcal{D}_\jointGrp\ham(\optimalHamState)+ \mathcal{D}_\jointGrp \Big(\sum_{i=1}^{\numPlants}\big\langle \ithMu,\stateConst(\varQ{t}, \varX{t})\big\rangle\Big)\Big)\bigg)\\
				\dualEucld_{t-1} = \mathcal{D}_\jointEucld\ham(\optimalHamState) +  \mathcal{D}_\jointEucld\Big(\sum_{i=1}^{\numPlants} \big\langle \ithMu, \stateConst(\varQ{t}, \varX{t})\big\rangle \Big); 
				\end{cases}
			\end{align*}
		\item \label{transversality} transversality:
			\begin{align*}
			\RhoTI[\nS-1] &= \cotanLift{ \ithId}\PhiI_{\kappaI(\opQT[\nS])}\bigg(\lambdaMap\Big(\NU \mathcal{D}_\jointGrp \ithCost[\nS](\varQ{\nS},\varX{\nS}) + \mathcal{D}_\jointGrp  \Big(\text{$\textstyle \sum_{i=1}^{\numPlants}$}\big\langle \ithMu[\nS], \stateConst[\nS](\varQ{\nS}, \varX{\nS})\big\rangle\Big)\Big) \bigg)\\ 
			\dualEucld_{\nS-1} &=\NU \mathcal{D}_\jointEucld\ithCost[\nS](\varQ{\nS},\varX{\nS}) + \mathcal{D}_\jointEucld\Big(\text{$\textstyle \sum_{i=1}^{\numPlants}$} \big\langle \ithMu[\nS], \stateConst[\nS](\varQ{\nS}, \varX{\nS})\big\rangle \Big); 
			\end{align*}
		\item \label{Ham NosPos Grad Cond}Hamiltonian maximization: \label{NonPosCond}
		\begin{align*}
			&\Big\langle \mathcal{D}_\jointCntrl\ham(\optimalHamState), \jointCntrl - \opUT \Big\rangle \leq 0 \quad \text{for all} \ \jointCntrl \belongsTo \suppCone{\ustar}{\opUT},
		\end{align*}
		where \(\suppCone{\ustar}{\opUT}\) is the support cone of \(\ustar\) with apex at \(\opUT\);
		\item \label{Comp Slackness Cond} complementary slackness:
		\begin{align*}
			 \ithMu \HadPr \stateConst(\varQ{t}, \varX{t})  = 0 \in \RPower[{\ithSysDimConstraint}] \quad \text{for all} \ \timeRangeFullSt \ \text{and} \ \iRange; 
		\end{align*}
		\item \label{NonPos Cond}non-positivity:
		\begin{align*}
			\vecCompare{\ithMu}{0} \quad \text{for all} \ \timeRangeFullSt \ \text{and} \ \iRange;
		\end{align*}
\item \label{MultiplexConst}multiplexing constraints: for the function \(z\) defined in \eqref{eq: constraintEq},
\begin{align*}
	\sum_{t=0}^{\nS-1} \auxDyn[\op{\uT}] = \begin{pmatrix}
	0 \\ 0
	\end{pmatrix}.
\end{align*}
	\end{enumerate}
\end{theorem}

\begin{remark}
	As convincingly argued in \cite{MR2009697}, the \dt{} dynamics \eqref{eq:ithSysDynamics} should be derived following the ideas of discrete mechanics to ensure greater numerical fidelity and accuracy; this particular technique ensures that the discretization does not violate the underlying manifold structure under time-discretization and also preserves important system invariants for conservative systems; consequently, it leads to greater accuracy than otherwise. Discrete mechanics is steadily becoming a popular tool to discretize the dynamics of physical systems; for instance, we refer the reader to \cite{Kobilarov} for examples of discretized dynamics of non-holonomic systems with symmetry, \cite{phogat2017discrete} for examples of discretized spacecraft attitude dynamics, \cite{InvertePendulum} for examples of discretized wheeled inverted pendula \footnote{Indeed, in \cite{Klauss} the authors implement a discrete mechanics based controller on a wheeled inverted pendulum. Video recording of one of the experiments conducted is available at \url{https://www.youtube.com/watch?v=Vw7vco-Rdrw&feature=youtu.be}}, \cite{nair2018discrete} for examples of discretized dynamics of interconnected mechanical systems, and \cite{nordkvist2010lie} for examples of discretized dynamics of rigid bodies evolving on the Lie group SE$(3)$.
\end{remark}

\begin{remark} 
	\ref{Ham NosPos Grad Cond} suggests that at the optimal point the Hamiltonian \eqref{eq: Hamiltonian Form} is non-increasing with respect to the control action along all the directions permissible by the support cone of set \(\ustar\) at \(\op{U}_t\). This behavior of the Hamiltonian function in the \dt{} context is weaker than the Hamiltonian maximization condition in continuous-time PMP (\cite[Theorem 22.26]{clark}, \cite[Chapter  4]{liberzon2011calculus}): in the continuous-time versions, the Hamiltonian is maximized over the entire admissible control action set at the optimal value of the control action. However, we have retained the name ``Hamiltonian maximization'' for \ref{Ham NosPos Grad Cond} to point out its connection to the standard continuous-time PMP.
\end{remark}

\begin{remark}
	The \dt{} PMP (Theorem \ref{th: 1}) for a system evolving on a matrix Lie group can be used, under certain conditions, to find explicit expressions of the control action in terms of the state and adjoint variables. One such condition is concavity of the Hamiltonian with respect to the control action as discussed in \cite{phogat2016discrete}. In most realistic cases, however, analytical solutions are difficult to arrive at. However, numerical algorithms can be deployed with the conditions of the PMP and an initial guess to find optimal state-action trajectories. For instance, in point to point ballistic reachability maneuvers, the discrete-time PMP yields a two point boundary value problem (TPBVP) that can be solved numerically via indirect methods such as single/multiple shooting, etc., in addition to direct optimization solvers. Indirect methods, as argued in \cite{Trelat2012}, are more accurate compared to direct method due to the extra information about the system dynamics provided by the necessary conditions of the PMP. Moreover, for systems evolving on high dimensional spaces, the indirect method turns out to be typically more efficient compared to the direct method. However, solving the multiplexed optimal control via indirect methods is non-trivial since it includes solving a combinatorial search problem; the development of such numerical methods is not the agenda of the article at hand.
\end{remark}

\begin{remark}
	In the preceding discussion we limited our scope to the situation where the server can provide control signals to at most one plant in the system at any time instant. However, this assumption can be relaxed to requiring that at most \(\mP\) plants out of \(\numPlants\) plants in the ensemble at any time instant may be controlled. The first order necessary conditions for optimality in such a situation can also be obtained in fashion similar to that of the case of \(\mP = 1\). Indeed, we define
	\begin{equation} 
	\begin{aligned}\label{eq: constraintEq2}
		& \ustar \ni \jointCntrl \mapsto \multiAuxDyn[\jointCntrl] \Let \\
		& \quad\sum_{\indexSet[] \in \textsf{Com}(\numPlants,\mP+1)} \bigl\|\uSig[{\indexSet[1]}]\bigr\|^2  \bigl\|\uSig[{\indexSet[2]}]\bigr\|^2 \dots  \bigl\|\uSig[{\indexSet[\mP+1]}]\bigr\|^2 \begin{pmatrix}
		1 \\1 
		\end{pmatrix} + \uSig[{\indexSet[1]}]\opr \uSig[{\indexSet[2]}] \dots \opr \uSig[{\indexSet[\mP+1]}] \begin{pmatrix}
		1 \\ -1
		\end{pmatrix} \in \R^2,
	\end{aligned}
\end{equation}
	where \(\uSig[\indexSet] \in \ithAdmControl[{\indexSet}]\) and \(\combi[\numPlants,\mP+1]\) denotes a set of all the combinations of  elements of the set \([\numPlants]^{\ast}\) containing \((\mP+1)\) distinct elements in them; obviously, \(\card{{\combi[\numPlants,\mP+1]}} = {\numPlants \choose \mP+1}\). Replacing the function \(\auxDyn[\jointCntrl]\) with \(\multiAuxDyn[\jointCntrl]\) in the definition of the Hamiltonian \eqref{eq: Hamiltonian Form} and appealing to Theorem \ref{th: 1}, a set of first order necessary conditions for optimality in this new situation may be easily obtained.
\end{remark}

Apart from the optimal control problem formulated in \eqref{eq:ProbFormulation} there is one special case that frequently arises in optimal control literature --- that of, optimal point-to-point ballistic reachability maneuvers. The precise problem statement is as follows: 
\begin{equation}
\label{eq:ProbFormulation_Corollary}
\begin{aligned}
\minimize_{(\uT)_{t=0}^{\nS-1}} &&& \mathcal{C}(\qHat, \xHat, \uHat) := \bigg(\sum_{t=0}^{\nS-1} \ithCost(\qT[t], \xT[t], \uT[t]) + \ithCost[\nS](\qT[\nS], \xT[\nS])\bigg) \\ 
\text{subject to} &&&
\begin{cases}
\text{dynamics \eqref{eq:ithSysDynamics} } &
\text{for all} \  \timeRange \ \text{and} \ \iRange, \\
\uT \belongsTo \uhash & \text{for all} \ \timeRange,\\
\stateConst \big(\grpIndexT, \stateIndexT\big) \leq 0 & \text{for all} \  \timeRangeSt \ \text{and} \ \iRange, \\

\left(\qT[0],\xT[0]\right)=\big(\ob{\qT[0]},\ob{\xT[0]}\big),\\
\left(\qT[\nS],\xT[\nS]\right)=\big(\ob{\qT[\nS]},\ob{\xT[\nS]}\big).\\
\end{cases} 
\end{aligned}
\end{equation}

\begin{corollary}\label{Corollary1}
	Let \(\uOptSeq\) be an optimal control sequence that solves \eqref{eq:ProbFormulation_Corollary} and let \(\stateTraj\) be the corresponding state trajectory. We define the Hamiltonian, for \(\NU \in \{-1,0\}\) and \(\auxAdj \in \RPower[2]\), by
	\begin{equation}
	\begin{aligned}\label{eq: Hamiltonian Form3}
	& \hamdef \ni  \hamvar \mapsto  \\ 
	& \ham{\left(\tau,\dualGrp,\dualEucld,\jointGrp,\jointEucld, \jointCntrl\right)} \coloneqq \NU\ithCost[\tau](\qT[],\xT[],\uT[]) + \sum_{i=1}^{\numPlants} \bigg\{ \Big \langle \iththeta, \ithExpMap^{-1}\big(\jumpMapIndexT[\tau](\grpIndexT[],\stateIndexT[])\big)\Big \rangle_{\ithLieAlg[i]}  \\ & + \Big\langle \ithxi,\ithSysDyn[\tau](\grpIndexT[],\stateIndexT[],\contIndexT[]) \Big\rangle  \bigg\} + \big \langle \auxAdj,  \auxDyn[{\uT[]}]\big \rangle \in \R, \\
	\end{aligned}
	\end{equation}
	where \(\dualGrp \defas \big(\iththeta[i]\big)_{i=1}^{\numPlants}, \dualEucld \coloneqq \big(\ithxi[i]\big)_{i=1}^{\numPlants}, \qT[] = \big(\grpIndexT[]\big)_{i=1}^{\numPlants} , \xT[] = \big(\stateIndexT[]\big)_{i=1}^{\numPlants}\), and \(\ithxi \belongsTo \big(\eucldStateSpace\big)\dual\) and \(\iththeta \belongsTo \ithLieAlg\dual\).  For \(\timeRange\) we define the transformation  
	\begin{align*}
	\ithLieAlg\dual \ni \iththeta[i]_t \mapsto \RhoTI \Let \Big(\mathcal{D}\ithExpMap^{-1}\big((\varQ{t})^{-1}\varQ{t+1}\big) \circ \tanLift{\ithId}{\PhiI_{(\varQ{t})^{-1}\varQ{t+1}}}\Big)\dual\big(\iththeta[i]_t\big)\in \ithLieAlg\dual \quad \text{for }\iRange,
	\end{align*}
	and denote the extremal lift of the state-action trajectory \((\opQT, \opXT, \opUT)\) under the optimal control \(\opUT\) at each time instant \(t\) by 
	\begin{align*}
	\optimalHamState \defas (t, \dualGrp, \dualEucld, \opQT, \opXT, \opUT).
	\end{align*} 
Then there exist an adjoint trajectory \(\big(\grpTheta, \grpXi\big)_{t=0}^{\nS-1} \subsetTo  \grpConfMfdDual\), covectors \(\big(\ithMu\big)_{t=1}^{\nS-1} \ \subsetTo \ \big(\eucldConstraintSpace\big)\dual\) for \(\iRange\),	such that the following conditions hold: 
	\begin{enumerate}[leftmargin=*, label={\rm (CL-\roman*)}, widest=iii, align=left]
		\item \label{Non-triviality conds} \ref{Non-triviality cond} holds;
		\item \label{StateAdjointHold}\ref{stateNDadjointDynamics} holds;
		\item \label{NonPosCondHold} \ref{Ham NosPos Grad Cond} holds; 
		\item \label{Comp Slackness CondCor} complementary slackness:
			\[
				\ithMu \HadPr \stateConst(\varQ{t}, \varX{t}) = 0 \in \RPower[{\ithSysDimConstraint}] \quad \text{for all}  \ \timeRangeSt \ \text{and} \ \iRange;
			\]
		\item \label{NonPos CondCor}non-positivity:
			\[
				\vecCompare{\ithMu}{0}  \quad \text{for all} \ \timeRangeSt \  \text{and} \ \iRange;
			\]
		\item \label{multiplex_Again}\ref{MultiplexConst} holds.
	\end{enumerate}
\end{corollary}

	\section{Proofs of Theorem \ref{th: 1} and Corollary \ref{Corollary1}}
	\label{sec: Proof}
\subsection{Proof of Theorem \ref{th: 1}}
We start by transforming the problem \eqref{eq:ProbFormulation} into one that can be solved using an existing PMP on matrix Lie groups \cite{phogat2016discrete}. 

Recall from \S\ref{sec: Problem Formulation} that \(\grpIndexT[t] \in \grpMfdIndex , \stateIndexT[t] \in \eucldStateSpace\) are the group and Euclidean state variables, respectively, corresponding to the \(\ith \) plant in the ensemble at the time instant \(t\).  For \(\timeRangeFull \) we define \(\qT[t] \Let \big(\grpIndexT[t] \big)_{i=1}^{\numPlants}\) and \(\xT[t] \defas \big(\stateIndexT[t]\big)_{i=1}^{\numPlants} \).  In view of the direct product group structure and \eqref{eq:ithSysDynamics}, for the part of the dynamics evolving on the matrix Lie group \(\gTot \), we write
\begin{align*}
	\qT[t+1] & = \big(\grpIndexT \jumpMapIndexT(\grpIndexT, \stateIndexT)\big)_{i=1}^{\numPlants} \quad \text{(the finite sequence)} \\
	&= \qT \jumpTot(\qT, \xT),
	\end{align*}
	where \(\jumpTot( \qT, \xT) \defas \big(\jumpMapIndexT(\grpIndexT, \stateIndexT)	\big)_{i=1}^{\numPlants}\), and for part of the dynamics evolving on the Euclidean space, we write
	\begin{align*}
		\xT[t+1] = \dynTot(\qT, \xT, \uT),
	\end{align*}	
	where \( \uT[t] \defas \big(\contIndexT \big)_{i=1}^{\numPlants}\) and \(\dynTot(\qT, \xT, \uT) \defas \big(\ithSysDyn(\grpIndexT, \stateIndexT, \contIndexT)\big)_{i=1}^{\numPlants}\). Thus, the dynamics of the joint control system can be concisely written as 
\begin{equation}\label{eq: jointCntrlSys}
	\begin{aligned}
		\qT[t+1] &= \qT[t]\jumpTot[t](\qT[t], \xT[t]), \\
		\xT[t+1] &= \dynTot[t](\qT, \xT, \uT).
	\end{aligned}
\end{equation}

 Next we reconfigure the multiplexing constraint. Recall that the multiplexing constraint is  naturally implied by the ``star''-shaped admissible action set \(\uhash\) (defined in \eqref{eq: AdmissibleActionSet}). We claim that
 \begin{align}\label{eqn: ustar}
  	\uhash = \Big\{\uT[] \in \ustar \ \defas \big( \ithAdmControl[1] \times \ithAdmControl[2] \times \dots \times \ithAdmControl[\numPlants] \big) \Big\lvert \ \ \auxDyn[{\uT[]}]= (0, 0)^\top  \Big\},
 \end{align}
where \(\ustar \ni \uT[] \mapsto \auxDyn\) is the map defined in \eqref{eq: constraintEq}; this equality of sets is immediate from Lemma \ref{rem: IFFRemark}. 
We introduce an auxiliary variable \(\aux[] \in \R^2\) and the dynamical system
\begin{align}\label{additional dynamics}
\aux[t+1] = \aux[t] + \auxDyn[\uT], \quad \aux[0] = (0,0)^\top,
\end{align}
\(\timeRange \), where \(\uT \in \ustar\) for all \(\timeRange \). It is immediate from \eqref{additional dynamics} and the additional terminal constraint \(\aux[N] = (0,0)^\top \) that \(\sum_{t=0}^{\nS-1}\auxDyn[\uT] = (0,0)^\top \), which, along with Lemma \ref{lem: zProp} guarantees that \(\auxDyn[\uT] = (0, 0)^\top\) for all \(\timeRange\). 

We claim that the following optimal control problem is equivalent to \eqref{eq:ProbFormulation}:
\begin{equation}
\label{eq:ProbFormulation_new}
\begin{aligned}
\minimize_{\left(\uT\right)_{t=0}^{\nS-1}} &&& \mathcal{C}(\qHat,\xHat,\uHat) \defas \sum_{t=0}^{\nS-1} \ithCost\big(\qT, \xT, \uT\big) + \ithCost[\nS]\big(\qT[\nS], \xT[\nS]\big) \\ 
\text{subject to} &&&
\begin{cases}
\text{dynamics \eqref{eq: jointCntrlSys} } &
\text{for all} \
\timeRange,  \\
\text{dynamics \eqref{additional dynamics}} &
\text{for all} \ \timeRange,\\
\uT \belongsTo \ustar & \text{for all} \ \timeRange, \\
\stateConst \bigl(\grpIndexT, \stateIndexT\bigr) \leq 0 & \text{for all} \
\timeRangeFullSt\ \text{and} \ \iRange, \\
\left(\qT[0],\xT[0], \aux[0]\right)=\bigl(\ob{\qT[0]},\ob{\xT[0]},(0,0)^\top\bigr), \\ \aux[N] = (0,0)^\top,
\end{cases} 
\end{aligned}
\end{equation}
where \(\qHat \defas \big(\qT\big)_{t=0}^{\nS}, \xHat \defas \big(\xT\big)_{t=0}^{\nS}\) and \(\uHat \defas \big(\uT\big)_{t=0}^{\nS-1}\). We shall establish this equivalence of the two problems \eqref{eq:ProbFormulation} and \eqref{eq:ProbFormulation_new} below. For the moment we observe that the first order necessary conditions for solution of \eqref{eq:ProbFormulation_new} are given by Theorem \ref{thm:DMP} in Appendix \ref{app: DPMP}. Indeed, by augmenting the Euclidean state variable \(\xT\) with the variable \(\aux[t]\) and correspondingly the adjoint variable \(\dualEucld_t\) with \(\auxAdj_t\),  appealing to Theorem \ref{thm:DMP} for the resulting optimal control problem formulated in \eqref{eq:ProbFormulation_new}, we get the first order necessary conditions for optimality presented in Theorem \ref{prop: ApplyDMP} below. Towards the end of this subsection we shall establish a connection between Theorem \ref{th: 1} and Theorem \ref{prop: ApplyDMP}, in particular that the necessary conditions of Theorem \ref{prop: ApplyDMP} imply those in Theorem \ref{th: 1}. Recall from \S\ref{sec: Problem Formulation} that the joint control system evolves on \(\gTot\times\RPower[\sysDimState] \), and that the Lie algebra of the matrix Lie group \(\gTot \) is denoted by \(\liea \).
\begin{theorem}\label{prop: ApplyDMP}
	Let \(\uOptSeq\) be an optimal control trajectory that solves \eqref{eq:ProbFormulation_new} and let \(\stateTraj\) be the corresponding state trajectory. We define, for \(\NU \in \{-1,0\}\),  the Hamiltonian by
	\begin{equation}
	\begin{aligned}\label{eq: Hamiltonian Form22}
	& [\nS-1] \times \liea^* \times \big(\R^{\sysDimState}\big)^* \times \big(\R^2\big)\dual \times\gTot \times \R^{\sysDimState} \times \big(\R^2\big)\times \R^{\sysDimControl} \ni  \big( \tau,\dualGrp,\dualEucld, \auxAdj, \qT[], \xT[],\auxVAR, \uT[] \big) \\ 
	& \quad \mapsto  \hamkp{\left( \tau,\dualGrp,\dualEucld, \auxAdj, \jointGrp, \jointEucld, \auxVAR, U \right)} \coloneqq \NU\ithCost[\tau](\qT[],\xT[],\uT[]) +  \big \langle \dualGrp, \ithExpMap[{\gTot}]\big(\jumpTot[\tau](\qT[], \xT[])\big)\big \rangle_{\liea}\\
		& \qquad + \big\langle \dualEucld, \dynTot[\tau](\qT[], \xT[], \uT[])\big \rangle  + \big \langle \auxAdj,  \auxVAR + \auxDyn[{\uT[]}]  \big \rangle \in \R.
	\end{aligned}
	\end{equation}
	For \(\timeRange\) we define the transformation
	\begin{align*}
	\liea^*\ni \dualGrp_t \mapsto \RhoTIFull \Let \Big(\mathcal{D}\ithExpMap[{\gTot}]\big((\op{\jointGrp}_t)^{-1}\op{\jointGrp}_{t+1}\big) \circ \tanLift{\IdEle[]}\leftAction_{(\op{\jointGrp}_{t})^{-1}\op{\jointGrp}_{t+1}}\Big)\dual\big(\dualGrp_t\big) \in \liea^*. 
	\end{align*}
	We denote the extremal lift of the state-action trajectory \((\opT{\jointGrp}, \opT{\jointEucld}, \opT{\jointCntrl}) \) under the optimal control \(\opT{\jointCntrl}\) at each time instant \(t\) by
	\begin{align*}
	\optimalHamState \defas (t, \dualGrp, \dualEucld, \opT{\jointGrp}, \opT{\jointEucld}, \opT{\jointCntrl}).
	\end{align*}  
	Then there exist an adjoint trajectory \(\big(\grpTheta, \grpXi\big)_{t=0}^{\nS-1} \subsetTo \grpConfMfdDual\), and a trajectory \(\big(\ithMu\big)_{t=1}^{\nS} \subsetTo \big(\eucldConstraintSpace\big)\dual\) for \(\iRange\) such that the following conditions hold: 
	\begin{enumerate}[leftmargin=*, label={\rm (JMP-\roman*)}, widest=iii, align=left]
		\item \label{Non-triviality cond22} non-triviality:
		the adjoint variables \((\grpTheta, \grpXi)\) for all \(\timeRange\), the covectors \(\ithMu\) \((\)for all \(\timeRangeFullSt\) and \(\iRange )\), the scalar \(\NU\), and the vector \(\auxAdj\) do not vanish simultaneously;
		
		\item \label{stateNDadjointDynamics22} state and adjoint system dynamics:
		\begin{align*}
			\text{states} & \begin{cases}
		\op{\jointGrp}_{t+1}  =\op{\jointGrp}_{t}  \exp_{\gTot}\big(\mathcal{D}_{\dualGrp}\hamkp(\optimalHamState)\big) \\
		\op{\jointEucld}_{t+1} = \mathcal{D}_{\dualEucld}\hamkp(\optimalHamState)\\
		\op{\auxVAR}_{t+1} = \mathcal{D}_{\auxAdj}\hamkp(\optimalHamState)
		\end{cases}  \\
			\text{adjoints} & \begin{cases}
		\RhoTIFull[{t-1}] =  \coAd{\ithExpMap[{\gTot}]\big(-\mathcal{D}_{\dualGrp}\hamkp(\optimalHamState)\big)}	\RhoTIFull[{t}] + \\ \hspace{2cm} \cotanLift{\IdEle[]}{\leftAction_{{\opT{\jointGrp}}}}\bigg(\mathcal{D}_{\jointGrp}\hamkp(\optimalHamState)+ \mathcal{D}_\jointGrp\Big(\sum_{i=1}^{\numPlants} \big\langle \ithMu,\stateConst(\varQ{t}, \varX{t})\big\rangle\Big)\bigg)\\ 
		\dualEucld_{t-1} = \mathcal{D}_{\jointEucld}\hamkp(\optimalHamState) +  \mathcal{D}_\jointEucld\Big(\sum_{i=1}^{i=\numPlants} \big\langle \ithMu, \stateConst(\varQ{t}, \varX{t})\big\rangle \Big)\\
		\auxAdj_{t-1} = \mathcal{D}_{\auxVAR} \hamkp(\optimalHamState) +\mathcal{D}_{\auxVAR} \Big(\sum_{i=1}^{i=\numPlants} \big\langle \ithMu, \stateConst(\varQ{t}, \varX{t})\big\rangle\Big),
		\end{cases} 
		\end{align*}
		where \(\leftAction\) is the left action on the Lie group \(\gTot\);
		
		\item \label{transversality22} transversality:
		\begin{align*}
		\Theta_{\nS-1} &= \cotanLift{ \IdEle[]}{\Phi_{\op{\jointGrp}_\nS}}\bigg(\NU \mathcal{D}_\jointGrp  \ithCost[\nS](\op{\jointGrp}_{\nS}, \op{\jointEucld}_{\nS}) +  \mathcal{D}_\jointGrp  \Big(\text{$\textstyle \sum_{i = 1}^{\numPlants}$}\big\langle \ithMu[\nS], \stateConst[\nS](\grpIndexT[\nS-1], \varX{\nS})\big\rangle\Big)\bigg)\\ 
		\dualEucld_{\nS-1} &= \mathcal{D}_{\jointEucld}\hamkp(\optimalHamState[\nS]) + \mathcal{D}_{\jointEucld}\Big( \text{$\textstyle \sum_{i = 1}^{\numPlants}$}\big\langle \ithMu[\nS], \stateConst[\nS](\grpIndexT[\nS], \varX{\nS})\big\rangle \Big); 
		\end{align*}
		
		\item \label{Ham NosPos Grad Cond22}Hamiltonian maximization: \label{NonPosCond22}
		\begin{align*}
		&\Big\langle \mathcal{D}_{\uT[]}\hamkp(\optimalHamState), \uT[] - \op{\uT} \Big\rangle \leq 0 \quad \text{for all} \ \uT[] \belongsTo \suppCone{\ustar}{\op{\uT}},
		\end{align*}
		where \(\suppCone{\ustar}{\op{\uT}}\) is the support cone of \(\ustar\) with the apex at \(\op{\uT}\);
		
		\item \label{Comp Slackness Cond22} complementary slackness:
		\begin{align*}
		 \ithMu \HadPr \stateConst(\varQ{t}, \varX{t})  = 0 \in \RPower[{\ithSysDimConstraint}] \quad \text{for all} \ \timeRangeFullSt \ \text{and} \ \iRange; 
		\end{align*}
		\item \label{NonPos Cond22}non-positivity:
		\begin{align*}
		\vecCompare{\ithMu}{0}   \quad \text{for all} \ \timeRangeFullSt \ \text{and} \ \iRange;
		\end{align*}
		\item \label{MultiplexConst22}multiplexing constraint:
		\begin{align*}
		\sum_{t=0}^{\nS-1} \auxDyn[\op{\uT}] = \begin{pmatrix}
		0 \\ 0
		\end{pmatrix}.
		\end{align*}
	\end{enumerate}
\end{theorem}

The condition \ref{MultiplexConst22}, in Theorem \ref{prop: ApplyDMP} is due to the constraint \(\aux[0] = \aux[\nS] = (0,0)\top\).
%

\begin{remark}\label{rem: ChiEvo}
	Note that the adjoint variable \(\auxAdj_t\) remains constant with respect to time. This follows immediately from the evolution of the adjoint \(\auxAdj_t\) described in \ref{stateNDadjointDynamics22}. Indeed, for \(\timeRange\) we have
	\begin{equation}
	\begin{aligned}
		\auxAdj_{t-1} &= \mathcal{D}_{\auxVAR} \hamkp(\optimalHamState) +\mathcal{D}_{\auxVAR} \bigg(\sum_{i=1}^{i=\numPlants} \big\langle \ithMu, \stateConst(\varQ{t}, \varX{t})\big\rangle\bigg) \\
		           &= \auxAdj_t,
	\end{aligned}  
	\end{equation}
	which shows that \(\auxAdj \defas \auxAdj_0 = \auxAdj_1 = \dots = \auxAdj_{\nS-1}\). 
\end{remark}
In view of Remark \ref{rem: ChiEvo} we can rewrite the Hamiltonian \eqref{eq: Hamiltonian Form22} as the one mentioned in \eqref{eq: Hamiltonian Form}. Note that we have removed the dependence of the variable \(\auxVAR\) on the Hamiltonian \eqref{eq: Hamiltonian Form} because the variable \(\auxVAR\) is not used in obtaining the necessary conditions presented in Theorem \ref{th: 1}, and is therefore redundant.

We return to the topic of establishing a connection between Theorem \ref{th: 1} and Theorem \ref{prop: ApplyDMP} and to this end we utilize some properties of direct product of matrix Lie groups presented in Appendix \ref{app: DiPr}. This connection between the necessary conditions in the two theorems is established by observing the following:  
	\begin{enumerate}[leftmargin=*, label={\rm ({Eqv}-\roman*)}, widest=iii, align=left]
		\item \ref{Non-triviality cond} is identical to \ref{Non-triviality cond22}.
		
		\item \label{proof: adj} 
		Using Lemma \ref{lem: expdist} and \ref{stateNDadjointDynamics22} we see that the evolution on each of the matrix Lie groups \(\grpMfdIndex\) is as described in \ref{stateNDadjointDynamics}. The evolution of the state variable on the Euclidean space in \ref{stateNDadjointDynamics} is the same as that in \ref{stateNDadjointDynamics22}. Under the natural identification of the (dual) Lie algebra of the joint matrix Lie group \(\gTot\) with the direct sum of the (dual) Lie algebras of the constituent Lie groups, for \(\RhoTIFull \in \liea\dual\) we can find \(\RhoTI \in \ithLieAlg\dual\) for \(\iRange\), such that \(\RhoTIFull = \big(\RhoTI\big)_{i=1}^{\numPlants}\) in the notation of Theorem \ref{th: 1}. 
	Moreover, from Lemma \ref{lem: expdist} and Lemma \ref{Prop: coadjoint} we have 
		\begin{align}\label{eq: Adjoint Lift}
			\coAd{\ithExpMap[{\gTot}](-\mathcal{D}_{\dualGrp}\ham(\optimalHamState))} \RhoTIFull = \Big(\coAd{\exp_{\grpMfdIndex}(-\ithProjMapPi(\mathcal{D}_{\dualGrp}\ham(\optimalHamState)))} \RhoTI \Big)_{i=1}^{\numPlants},
		\end{align}
		where $\ithProjMapPi$ is the map defined towards the end of \S \ref{subsec:firstMR}. From Lemma \ref{lem: TanMapDist} we can conclude that the cotangent maps also split in the preceding fashion. In other words, the adjoint dynamics of \(\RhoTIFull\) in \ref{stateNDadjointDynamics} is equivalent to that of \ref{stateNDadjointDynamics22}. Furthermore, the dynamics of the adjoint variable \(\dualEucld_t\) in \ref{stateNDadjointDynamics} is the same as that in \ref{stateNDadjointDynamics22}.
		
		\item Using arguments similar to the ones presented in \ref{proof: adj} above, the equivalence between \ref{transversality} and \ref{transversality22} can be established.
		
		\item \ref{Ham NosPos Grad Cond}-\ref{MultiplexConst} are identical to \ref{Ham NosPos Grad Cond22}-\ref{MultiplexConst22}. 
	\end{enumerate}
	 
	 To complete the proof we must ensure that the admissible processes in the problems  \eqref{eq:ProbFormulation} and \eqref{eq:ProbFormulation_new} are in bijective correspondence, a fact that is almost immediate. Indeed, note that for any trajectory \(\big(\qT, \xT, \uT\big)_{t=0}^{\nS}\) that is admissible in \eqref{eq:ProbFormulation}, the trajectory \(\big(\qT, \xT, \uT, 0\big)_{t=0}^{\nS}\) is admissible in \eqref{eq:ProbFormulation_new}. Similarly, for any trajectory \(\big(\qTnew, \xTnew, \uTnew, \wT\big)_{t=0}^{\nS}\) that is admissible in \eqref{eq:ProbFormulation_new}, the trajectory \(\big(\qTnew, \xTnew, \uTnew\big)_{t=0}^{\nS} \) is admissible in \eqref{eq:ProbFormulation} because any admissible trajectory of \eqref{eq:ProbFormulation_new} has \(\wT =0 \) for all \(\timeRange \), which forces the control action \(\uTnew \in \uhash\) for all \(\timeRange\).   Thus, the admissible sets of to \eqref{eq:ProbFormulation} and \eqref{eq:ProbFormulation_new} are in bijective correspondence. 
	 \qed
	 \begin{remark}\label{rem: Convexify}
	 	Note that one cannot directly lift the result Theorem \ref{thm:DMP} to solve the optimal control problem \eqref{eq:ProbFormulation} since one of the requirements of Theorem \ref{thm:DMP} is the convexity of the admissible control action set. It is because of this reason that we first modified the original problem \eqref{eq:ProbFormulation} to \eqref{eq:ProbFormulation_new}. We enumerate some important properties of the two characterizations of the admissible joint control action set: 
	 		\begin{enumerate}[label=(\alph*), leftmargin=*, widest=b, align=left]
	 			\item \(\uhash\) is a non-convex set while \(\ustar\) is a convex set. Convexity of \(\ustar\) is attributed to the convexity of \(\ithAdmControl[i]\) for all \(\iRange\) as per Assumption \ref{ass: Existence} and the fact that finite Cartesian products of convex sets yield in a convex sets \cite[Chapter 1]{fundamentalsCA}. 
	 			\item The auxiliary dynamics \eqref{additional dynamics} along with the additional terminal constraint \(\aux[N]= (0,0)^\top \) encodes the multiplexer constraint; i.e., at any instant the control action resides on one of the branches of the ``star''-shaped set \uhash. (refer Appendix \ref{app: app_z})  
	 		\end{enumerate}
	 \end{remark}
	 \begin{remark}
	 	Note that we implemented the multiplexer constraint by introducing an auxiliary variable. In \cite{KPPCP} and \cite{paruchuri2017discrete} authors have implemented frequency constraints on the control action in the similar fashion. The auxiliary system dynamics in such cases is a manifestation of the constraints enforced on the discrete Fourier transform of the control trajectories \cite{paruchuri2017discrete}. 
	 \end{remark}
\subsection{Proof of Corollary \ref{Corollary1}} 
 Observe that the optimal control problem \eqref{eq:ProbFormulation_Corollary} is a special case of \eqref{eq:ProbFormulation} in the sense that we can view the constraints imposed on the final states as state-inequality constraints as presented in \cite[Appendix A1-A2]{phogat2016discrete}. The proof now follows by extending Theorem \ref{th: 1} to the case where the final states of the joint system are constrained to lie on a submanifold, say \(\mfdFinal\). This idea can be formalized by representing the end point constraints \((\qT[N], \xT[N]) \in \mfdFinal\) as a state inequality constraint as presented in \cite[Appendix A1]{phogat2016discrete}. Point-to-point ballistic reachability maneuvers are then special cases of this problem in the sense that manifold \(\mfdFinal\) is reduced to a singleton set \cite[Appendix A2]{phogat2016discrete}. The end point constraints in such cases imply that the transversality conditions are trivially satisfied; consequently, transversality conditions do not appear in the statement of Corollary \ref{Corollary1} \qed 

	\section{Numerical Experiment}
	\label{sec: Num Exp}
	\subsection{Application to the attitude control of satellite}
Satellites in outer space are often commanded to perform orientation maneuvers about specified axes to point star sensors at some specific coordinates in deep space, pointing cameras in a specific desired direction for imaging purposes, to position solar panels for effective tracking of the sun for optimal energy harvesting, etc.; see, e.g., \cite{scrivener1994survey}. Throughout the duration of such commanded maneuvers there are strict limitations needed so that the motion of any satellite should strictly stay within certain limits, e.g., of momenta, speed, etc., to avoid mechanical failures. 
For an illustration of our results we consider a system of two satellites performing single-axis energy optimal maneuvers without violating pre-specified constraints on their control actions and angular momenta. Moreover, the control action commands are sent to these satellites via a single shared server, which imposes the limitation that the control signal can be dispatched to only one among the two satellites at any time instant. 
 
The configuration space of a satellite undergoing single-axis attitude motion is \(\SATmfd\), which is isomorphic to \(\ism\). We have deliberately not chosen a general rigid body orientation maneuver, with the configuration manifold \(\SATfull\), for the illustration at hand in order to get a better visualization of the results in the form of figures, while at the same time the coordinate-free nature of the controller is clearly amplifed. We borrow the discrete-time model of a satellite obtained using discrete mechanics from \cite{phogat2017discrete}:
\[
\begin{cases}
	\Rot{t+1}{i} = \Rot{t}{i} \SATjump\bigl(\Mom{t}{i}\bigr),\\
	\Mom{t+1}{i} = \Mom{t}{i} + h \torq{t}{i},
\end{cases} \quad \quad \text{for} \ \  i \in \{1, 2\},
\]
where \(\SATjump(\SATom) \defas \begin{pmatrix} \sqrt{1-h^2 \SATom^2} & - h \SATom \\  h \SATom & \sqrt{1-h^2 {\SATom^2}} \end{pmatrix}\), \(\Mom{t}{i} \in \R\) is the angular momentum, \(\Rot{t}{i} \in \so(2)\) is the rotation matrix, and \(\torq{t}{i} \in \R\) is the control action applied about the axis of rotation of the \(\ith\) satellite at time instant \(t\). Thus, the configuration space for this joint system of two satellites is \(\R^2 \times \so(2) \times \so(2) \simeq \R^2 \times S^1\times S^1\).

Fix \(\torqConstraint, \MomConstraint> 0\) for \(i \in \{1,2\}\). At each time instant \(t\) we enforce box constraints on the control action of the form \(\bigl|{\torq{t}{i}}\bigr| \leq \torqConstraint\) and on the angular momentum of the form \(\frac{1}{2}\big((\Mom{t}{i})^2 - \MomConstraint^2 \big) \leq 0\) of the \(\ith\) satellite. We formulate the problem as follows
\begin{equation}
\label{eq:exmpl}
\begin{aligned}
	& \minimize_{} && \mathscr{J} \left(\SATcont\right) \defas \sum_{i=1}^{2}\sum_{t=0}^{\nS-1} \frac{(\torq{t}{i})^2}{2} \\
	& \text{subject to} &&
\begin{cases}
	\left.
	\begin{aligned}
		& \left.
		\begin{aligned}
			& \Rot{t+1}{i} = \Rot{t}{i} \SATjump(\Mom{t}{i})\\
			& \Mom{t+1}{i} = \Mom{t}{i} + h \torq{t}{i} \\
			& \bigl|{\torq{t}{i}}\bigr| \leq \torqConstraint
		\end{aligned}\right\} & \text{for\;} \timeRange, \\
		& \tfrac{1}{2}\big((\Mom{t}{i})^2-\MomConstraint^2 \big)\leq 0 & \text{for\;} \timeRangeSt,\\
		& \bigl(\Rot{0}{i},\Mom{0}{i}, \aux[0] \bigr) = \Bigl(\Rot{\text{In}}{i}, \Mom{\text{In}}{i}, (0,0)\top \Bigr),\\  
		& \bigl(\Rot{\nS}{i},\Mom{\nS}{i}, \aux[\nS]\bigr) = \Bigl(\Rot{\text{Fi}}{i}, \Mom{\text{Fi}}{i}, (0,0)\top\Bigr),
	\end{aligned}\right\} & \text{for\;} i \in \{1, 2\},\\
	\aux[t+1] = \aux[t] + \begin{pmatrix}
		\torq{t}{1}\torq{t}{2}\big(\torq{t}{1}\torq{t}{2}+1\big) \\
		\torq{t}{1}\torq{t}{2}\big(\torq{t}{1}\torq{t}{2}-1\big)
		\end{pmatrix} \ \ \ \ \ \text{for\;} \timeRange,
\end{cases} 
\end{aligned}
\end{equation}
where the subscripts ``$\text{In}$'' and ``$\text{Fi}$'' denote the corresponding values at the initial and the final instants, respectively.

We obtain the first order necessary conditions for optimality of the preceding point-to-point ballistic reachability maneuvers by appealing to Corollary  \ref{Corollary1}. We define a map \(  \dlieaSO \ni t \mapsto \hat{t} \in \R\). To be more precise, in the preceding operation we identify \(\dlieaSO \) with \(\mathfrak{so}(2)\) and then use the vector space homeomorphism from \(2\times2 \) skew symmetric matrices to \(\R\). Given a scalar \(\NU \in \{-1,0\}\) and \(\SATauxAdj[] \defas (\SATauxAdj[1], \SATauxAdj[2]) \in (\R^2)\dual\), we define the Hamiltonian by
\begin{multline*}
	(\dlieaSO \times \dlieaSO) \times (\R^2)\dual \times (\so(2) \times \so(2)) \times \R^4 \ni (\SATdualg[], \SATdualE[], \SATg, \SATE[] , \SATu[]) \mapsto\\
	\ham[{\SATauxAdj[]}](\SATdualg[], \SATdualE[], \SATg, \SATom , \SATu[]) \defas \sum_{i=1}^{2}\bigg(\NU \frac{(\torq{}{i})^2}{2} + \big\langle\iSATdualg{}{i},  \ithExpMap[i]^{-1}(\SATjump(\Mom{}{i}))\big\rangle_{\mathfrak{so}(2)} + \iSATdualE{}{i}(\Mom{}{i} + h \torq{}{i})\bigg)\\ + \bigg\langle \SATauxAdj[], \begin{pmatrix} \torq{}{1}\torq{}{2}\big(	\torq{}{1}\torq{}{2}+1\big) \\ \torq{}{1}\torq{}{2}\big(	\torq{}{1}\torq{}{2}-1\big) \end{pmatrix}\bigg\rangle,
\end{multline*}
where \(\SATdualg[] \Let (\iSATdualg{}{1}, \iSATdualg{}{2}), \SATdualE[] \Let (\iSATdualE{}{1}, \iSATdualE{}{2}) , U \Let(\torq{}{1}, \torq{}{2})\) and \(\SATom \Let (\Mom{}{1}, \Mom{}{2})\). For \(\timeRange \) we define a transformation
\begin{align*}
\dlieaSO \ni \iSATdualg{t}{i} \mapsto \RhoTI[t] \Let \Big(\mathcal{D}\ithExpMap[]^{-1}\big(\SATjump(\opMom{t}{i})\big) \circ \tanLift{\ithId}{\PhiI_{\SATjump(\opMom{t}{i})}}\Big)\dual\big(\iSATdualg{t}{i}\big) \in \dlieaSO.   
\end{align*}

Let \(\big(\opSATu\big)_{t=0}^{\nS-1}\) be an optimal control that solves \eqref{eq:exmpl}. Corollary \ref{Corollary1} asserts that there exist a trajectory \(t \mapsto \optimalHamState \defas\Big( \big(\iSATdualg{t}{i},  \iSATdualE{t}{i} \big)_{i=1}^{i=2}, \big( \opRot{t}{i}, \opMom{t}{i} \big)_{i=1}^{i=2} \Big)\) on \((\mathfrak{so}(2)\dual \times \mathfrak{so}(2)\dual) \times (\R^2)\dual  \times (\so(2) \times \so(2)) \times \R^2\), a covector trajectory \(\bigl(\ithMu\bigr)_{t=1}^{\nS-1} \subsetTo (\R^2)\dual\) for \(i \in \{ 1,2\} \), a scalar \(\NU\), and a vector \(\SATauxAdj[] \in (\R^2)\dual\), such that:  
\begin{enumerate}[leftmargin=*, label=(V-\roman*), widest=b, align=left]
	\item \label{cond:v1} The non-triviality condition \ref{Non-triviality conds} holds, 
	\item The state and adjoint dynamics, for \(i \in \{1,2\}\) and \(t \in [N-1]\), are given by
	\begin{align*}
		\opRho{t-1}{i}&= \ad{\exp(-\mathcal{D}_{\iSATdualg{}{i}}\ham(\optimalHamState))}\dual \opRho{t}{i} = \SATjump(\opMom{t}{i}) \opRho{t}{i} \SATjump(\opMom{t}{i})^\top =  \opRho{t}{i},\\
		\opdualE{t-1}{i} &= \mathcal{D}_{\Mom{}{i}}\ham(\optimalHamState) + \ithMu\opMom{t}{i} = \frac{h \hatOpdualg{t}{i}}{\sqrt{1-h^2(\opMom{t}{i})^2}} + \opdualE{t}{i} + \ithMu \opMom{t}{i},
	\end{align*}
	Thus, the state and adjoint dynamics are given by
	\begin{equation}\label{eq:adjpmp}
		\text{adjoints} \begin{cases}
			\opRho{t-1}{i} = \opRho{t}{i},\\
			\opdualE{t-1}{i} =  \frac{h \hatOpdualg{t}{i}}{\sqrt{1-h^2(\opMom{t}{i})^2}} + \opdualE{t}{i} +\ithMu \opMom{t}{i},
		\end{cases} 
		\text{states} \begin{cases}
		\opRot{t+1}{i} = \opRot{t}{i} \SATjump(\opMom{t}{i})\ ,\\
		\opMom{t+1}{i}  =  \opMom{t}{i} + h \optorq{t}{i},\\
		\end{cases}
	\end{equation}
	
\item The Hamiltonian maximization condition for \(\NU = -1\) (i.e., the \emph{normal} case): We analyze the implications of \ref{NonPosCondHold} for the above system of two satellites. To this end, we first evaluate the directional derivative of the Hamiltonian function with respect to $\uT[]$ at the point \(\uT = (\torq{t}{1}, \torq{t}{2}) \in \R^2\), which is
		\begin{equation}\label{Hammax_ex}
		\mathcal{D}_{\uT[]} \ham(\cdot, \uT) = \begin{pmatrix}
			-\torq{t}{1} + h \Mom{t}{1} + 2 (\SATauxAdj[1]+\SATauxAdj[2]) \torq{t}{1} (\torq{t}{2})^2 + (\SATauxAdj[1]-\SATauxAdj[2]) \torq{t}{2}\\
			-\torq{t}{2} + h\Mom{t}{2} + 2 (\SATauxAdj[1]+\SATauxAdj[2]) (\torq{t}{1})^2 \torq{t}{2} + (\SATauxAdj[1]-\SATauxAdj[2])\torq{t}{1}
		\end{pmatrix}. 
		\end{equation} 
		We start with the assumption that at time \(t\) we have \(\multiplexFunc(t) = 1\) (that is, \(\optorq{t}{2} = 0\)); the case where \(\multiplexFunc(t) = 2\) follows similarly. Moreover, with a slight abuse of terminology we employ the term ``support cone at any instant'' to be the support cone of the admissible \emph{joint} control action set with apex at the optimal control action for that time instant. The following two cases arise: 
		\begin{itemize}[leftmargin=*]
		\item \textbf{Case I.}   \(\bigl|{\optorq{t}{1}}\bigr| < \torqConstraint[1]\): Since the control magnitude is within the prescribed bounds, the support cone is \(\R^2\). Thus, the Hamiltonian maximization condition \ref{NonPosCondHold} on substituting \eqref{Hammax_ex} gives us
		\begin{equation*}\label{eqCase1}
\begin{aligned}
\optorq{t}{1} &= h \opdualE{t}{1}, \\
h \opdualE{t}{2} &= (\SATauxAdj[1]-\SATauxAdj[2])\optorq{t}{1}.
\end{aligned}
		\end{equation*}
		
		\item \textbf{Case II.}  \(\optorq{t}{1} = \torqConstraint[1]\): This means that the control action constraint is active and the support cone is given by: \(\suppCone{\ustar}{{(\optorq{t}{1}, 0)}} = \big\{(\xvec[1], \xvec[2]) \in \R^2 | \xvec[1] \leq \torqConstraint[1] \big\}\). Appealing to \ref{NonPosCondHold} and substituting \eqref{Hammax_ex} we arrive at
			\begin{equation*}\label{eqCa2}
			\begin{aligned}
		 	h \opdualE{t}{1} &\ge \torqConstraint[1], \\
			h \opdualE{t}{2} &= (\SATauxAdj[1]-\SATauxAdj[2])\optorq{t}{1}.
			\end{aligned}
			\end{equation*}
		 Similarly, if \(\optorq{t}{1} = -\torqConstraint[1]\), we obtain the condition
			\begin{align*}
				h \opdualE{t}{1} & \leq -\torqConstraint[1],\\
				h \opdualE{t}{2} & = (\SATauxAdj[1]-\SATauxAdj[2])\optorq{t}{1}.
			\end{align*}
		\end{itemize}
	\item The Hamiltonian maximization condition for \(\NU = 0\) (i.e., the \emph{abnormal} case): The directional derivative of the Hamiltonian function with respect to \(U\) in this case is: 
	\begin{equation}\label{Hammax_ex2}
	\mathcal{D}_{\uT[]} \ham(\cdot, \uT) = \begin{pmatrix}
	h \Mom{t}{1} + 2 (\SATauxAdj[1]+\SATauxAdj[2]) \torq{t}{1} (\torq{t}{2})^2 + (\SATauxAdj[1]-\SATauxAdj[2]) \torq{t}{2}\\
h\Mom{t}{2} + 2 (\SATauxAdj[1]+\SATauxAdj[2]) (\torq{t}{1})^2 \torq{t}{2} + (\SATauxAdj[1]-\SATauxAdj[2])\torq{t}{1}
	\end{pmatrix}
	\end{equation} 
	for any control action \(\uT = (\torq{t}{1}, \torq{t}{2}) \in \R^2\). As in the case of \(\NU = -1\), we assume that \(\multiplexFunc(t) = 1\) (i.e., at time \(t\), \(\optorq{t}{2} = 0\)) without loss of generality. The following two cases arise:
	\begin{itemize}[leftmargin=*] 
		\item \textbf{Case I.}  \(\bigl|{\optorq{t}{1}}\bigr|< \torqConstraint[1]\): Since the control magnitude is within the prescribed bounds, the support cone is \(\R^2\), and the Hamiltonian maximization condition \ref{NonPosCondHold} upon substituting \eqref{Hammax_ex2} leads to
		\begin{equation*}\label{eqCase11}
		\begin{aligned}
		 \opdualE{t}{1} &= 0, \\
		h \opdualE{t}{2} &= (\SATauxAdj[1]-\SATauxAdj[2])\optorq{t}{1}.
		\end{aligned}
		\end{equation*}

		\item \textbf{Case II.} \(\optorq{t}{1} = \torqConstraint[1]\): The control action constraint is active, which leads to the support cone  \(\suppCone{\ustar}{{(\optorq{t}{1}, 0)}} = \big\{(\xvec[1], \xvec[2]) \in \R^2 | \xvec[1] \leq \torqConstraint[1] \big\}\). Appealing to \ref{NonPosCondHold} and substituting \eqref{Hammax_ex} we arrive at
		\begin{equation*}
		\begin{aligned}
		\opdualE{t}{1} &\ge 0, \\
		h \opdualE{t}{2} &= (\SATauxAdj[1]-\SATauxAdj[2])\optorq{t}{1}.
		\end{aligned}
		\end{equation*}
		Similarly, if \(\optorq{t}{1} = -\torqConstraint[1]\), we obtain the condition
			\begin{align*}
				\opdualE{t}{1} & \leq 0,\\
				h \opdualE{t}{2} & = (\SATauxAdj[1]-\SATauxAdj[2])\optorq{t}{1}.
			\end{align*}
	\end{itemize}
	\item \label{cond:vfinal} The complementary slackness condition and the non-positivity condition give us, in view of \ref{Comp Slackness CondCor}-\ref{NonPos CondCor}, that if the angular momentum constraint is inactive, then the corresponding covector \(\ithMu\) is zero; otherwise it is non-positive.
\end{enumerate}

For the purpose of numerical experiments we take help of the software \textsf{CasADi} \cite{Andersson2013b} where we use Interior Point Optimization (IPOPT) solver. 

The multiplexing constraint is implemented by introducing a new state variable \(\auxVAR \) with the dynamics given by \eqref{additional dynamics} and an additional terminal constraint \(\auxVAR_{\nS} = (0,0)^\top \). The results below correspond to a system of two satellites having moments of inertia \(\SATIner[1] = 800 \ {\si{\kilogram}\si{\meter}^2}\) and \(\SATIner[2] = 1200 \ {\si{\kilogram}\si{\meter}^2}\) about their respective axes of rotation. We report the following maneuver:
\begin{itemize}[leftmargin=*]
	\item \({\SATmaneuver[]:}\) Time of simulation = \(117\si{\second}\), step length = \(0.1\si{\second}\), \(\torqConstraint[1] = \torqConstraint[2] = 0.05\si{\newton}\si{\meter}, \MomConstraint[1] = \MomConstraint[2] = 0.1\si{\newton}\si{\meter}\si{\second}\);
	\begin{itemize}
		\item \(\left(\ang{0}{1},\Mom{0}{1}, \ang{0}{2},\Mom{0}{2} \right)= \left({0}{\degree},0  \si{\newton}\si{\meter}\si{\second} , {0}{\degree}, 0   \si{\newton}\si{\meter}\si{\second} \right)\);
		\item \(\left(\ang{N}{1},\Mom{N}{1}, \ang{N}{2},\Mom{N}{2} \right)= \left({45}{\degree}, 0.015  \si{\newton}\si{\meter}\si{\second} , {90}{\degree}, 0.015   \si{\newton}\si{\meter}\si{\second} \right)\).
	\end{itemize}	
	\end{itemize}	
	The numerical results for this maneuver are presented in Figure \ref{fig: omega_comparison}--\ref{fig: ControlComparisonText1}. {The parameters in the maneuver \(\SATmaneuver[]\) are obtained in a trial and error fashion so that the problem remains feasible. 
		
		As the timescale in \(\SATmaneuver[]\) is moderately large, to satisfy the end point requirements the optimal trajectory happens to take one complete `revolution' around the cylinder (see Figure \ref{fig: omega_comparison}) before attaining its final state. This brings out the benefit of using a geometric controller rather than a single-chart based controllers (such as the ones based on quaternions, Euler angles etc.), which suffer with the issues of singularity. 
		
		In order to complete one `revolution' the angular velocity has to remain almost constant around its peak for some time, which enforces relatively less chatter in control action during that phase (see Figure \ref{fig: ControlComparisonText1}). 
		
	} 
	

\begin{figure}[h]
	\centering
	\includegraphics[scale=1.1]{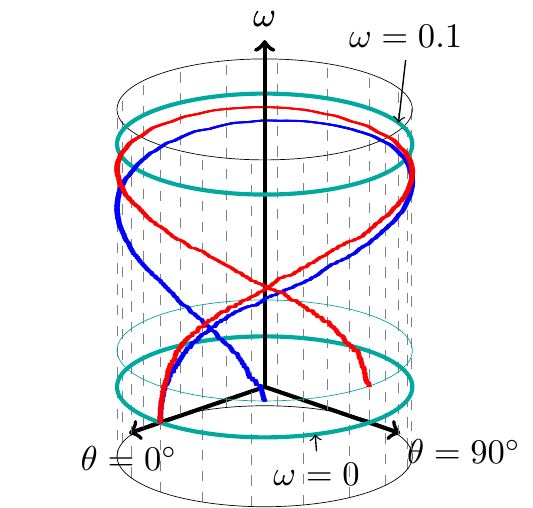}
	\caption{The blue trajectory on the cylinder shows the angular momentum of the satellite 1 and the red trajectory shows the angular momentum of satellite 2 for the maneuver $\SATmaneuver[]$. The animation is available here: \href{https://www.youtube.com/watch?v=EYULn5bG45A}{https://www.youtube.com/watch?v=EYULn5bG45A}} 
	\label{fig: omega_comparison}
\end{figure}
%
%
\begin{figure}[h]
	\centering
	\includegraphics[scale=0.4]{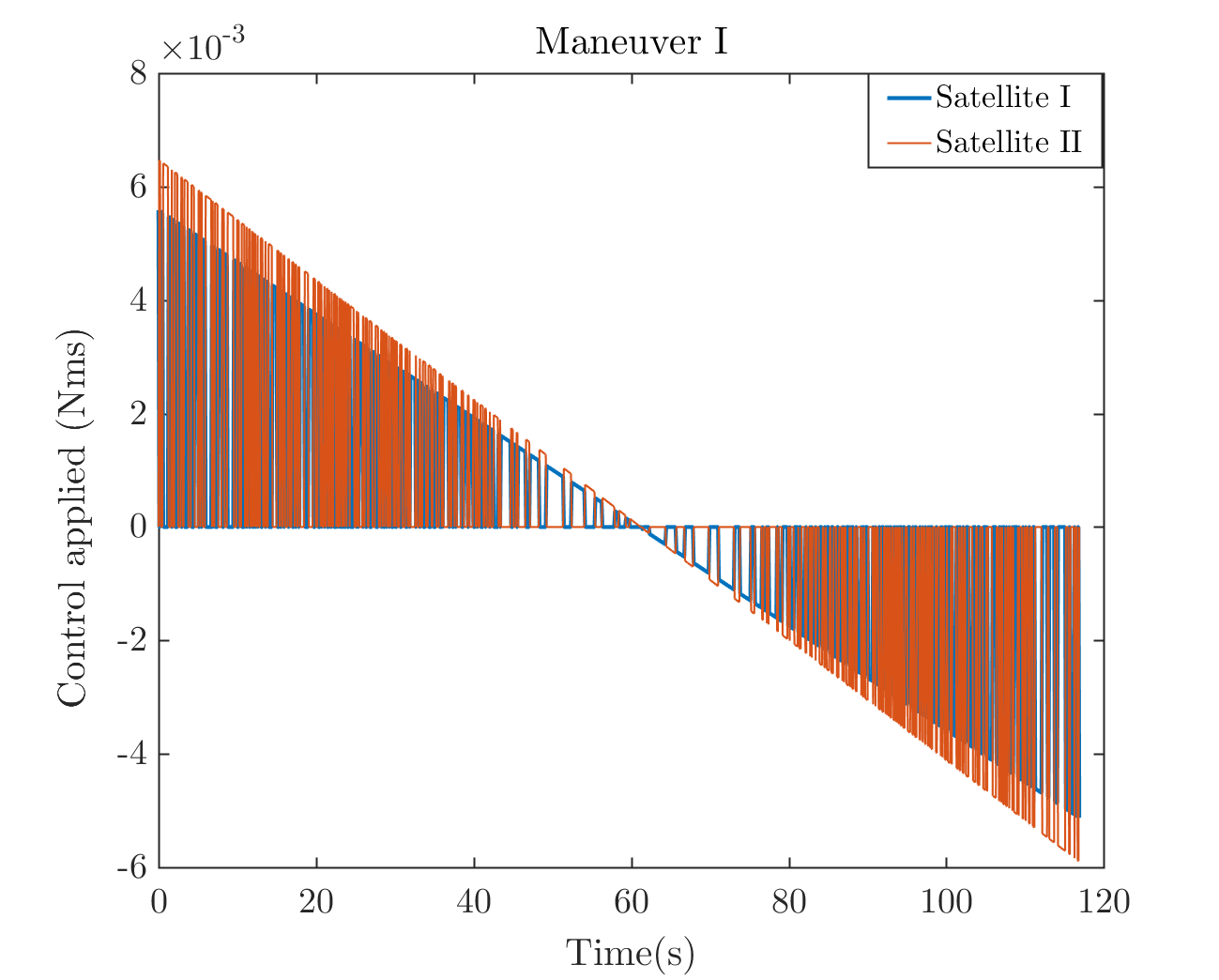}
	\caption{The figure illustrates the optimal control trajectory of both the satellites under maneuver $\SATmaneuver[]$. Note that the multiplexing constraint is followed at every instant of time. }
	
	\label{fig: ControlComparisonText1}
\end{figure}
The first order necessary conditions for optimality \ref{cond:v1}--\ref{cond:vfinal} in our experiments were verified against the numerical results for both of the maneuver \(\SATmaneuver[]\). In the process of obtaining the aforementioned simulation results, we received several outputs that do not satisfy the proposed necessary conditions, from the solver \(\textsf{CasADi}\), which were rejected for not satisfying the first order necessary conditions presented in this article. This brings out the need to develop a solver that utilizes the first order necessary conditions presented in this article to compute the solution trajectories. 

\subsection{Application to the control of underwater(UW) vehicle}

Another interesting setup where the system states evolve on non flat spaces is underwater (UW) vehicles. UW vehicles are becoming popular in research domains for their capabilities in conducting underwater surveillance \cite{shen2017trajectory}. We borrow the discrete time model for dynamics of the UW vehicle from \cite{nordkvist2010lie}. We have restricted the motion of system, for the sake of this article, on SE(2)\(\times\R^5\): 

\begin{equation}\label{eq: SE2dynamics}
\begin{aligned}
\begin{cases}
\Rot{t+1}{i} &= \Rot{t}{i} \SATjump\bigl(\Mom{t}{i}\bigr),\\
\Mom{t+1}{i} &= \Mom{t}{i} + h \TorqSE{t}{i},
\\
\dir{t+1}{i} &= \dir{t}{i}+h\Rot{t}{i}\vel{t}{i},\\
\Mass \vel{t+1}{i} &= \SATjump\bigl(\Mom{t}{i}\bigr)^\top \Mass \vel{t}{i} + h\Force{t}{i}
\end{cases} \quad \quad \text{for} \ \  i \in \{1, 2\}
\end{aligned}
\end{equation}
where \(\SATjump(\SATom) \defas \begin{pmatrix} \sqrt{1-h^2 \SATom^2} & - h \SATom \\  h \SATom & \sqrt{1-h^2 {\SATom^2}} \end{pmatrix}\), \(\Mom{t}{i} \in \R\) is the angular momentum, \(\Rot{t}{i} \in \so(2)\) is the rotation matrix, and \(\TorqSE{t}{i} \in \R\) is the torque applied about the axis of rotation, \(\Force{t}{i} \in \R^2 \) is the force vector applied on the center of mass, \(\dir{t}{i} \in \R^2 \) is the vector specifying position of the center of mass, \(\vel{t}{i} \in \R^2 \) is the velocity of the center of mass of the \(\ith\) UW vehicle at time instant \(t\). Thus, the configuration space for this joint system of two satellites is \(\R^{10} \times \se(2) \times \se(2)\). In this section from this point onwards, for any \(\delta \in \R^2 \) we shall use the notation \(\delta_x\) for the  first component and \(\delta_y \) for the second component of the vector.  

Fix \(\torqConstraint[ij], \MomConstraint[ij] >0\)  for \(i \in \{1,2\}, j \in \{1,2,3 \}\). At each time instant \(t\) we enforce control constraints of the form \(|{\TorqSE{t}{i}}| \leq \torqConstraint[i1], |(\Force{t}{i})_x|\leq \torqConstraint[i2], |(\Force{t}{i})_y|\leq \torqConstraint[i3]  \) and the state constraints on the angular and linear velocities of the form:
 \[\frac{1}{2}\big((\Mom{t}{i})^2 - \MomConstraint[i1]^2 \big) \leq 0, \frac{1}{2}\big(((\vel{t}{i})_x)^2 - \MomConstraint[i2]^2 \big) \leq 0, \frac{1}{2}\big(((\vel{t}{i})_y)^2 - \MomConstraint[i3]^2 \big) \leq 0\] of the \(\ith\) satellite. Similar to \eqref{eq:exmpl}, we formulate the optimal control problem as follows
\begin{equation}
\label{eq:exmpl2}
\begin{aligned}
& \minimize_{} && \mathscr{J} \left(\tau,\phi\right) \defas \sum_{i=1}^{2}\sum_{t=0}^{\nS-1} \frac{|\TorqSE{t}{i}|^2 + \|\Force{t}{i} \|_2^2}{2} \\
& \text{subject to} &&
\begin{cases}
& \text{dynamics} \ \eqref{eq: SE2dynamics} \\
& \text{state constraints}
\\
&\text{boundary conditions}
\\
& \text{multiplexing constraints \eqref{eqn: ustar2} }
\end{cases} 
\end{aligned}
\end{equation}

We refrain, here, from stating the necessary conditions for this example to limit the size of article. We directly present the simulation results which satisfy the necessary conditions for optimality. 

We perform the following maneuver for a pair of UW vehicle:
\begin{itemize}[leftmargin=*]
	\item \({\SATmaneuver[\star]:}\) Time of simulation = \(27\si{\second}\), step length = \(0.05\si{\second}\), \(\torqConstraint[11] = \torqConstraint[21] = 0.025\si{\newton}\si{\meter}, \torqConstraint[12] = \torqConstraint[13] = \torqConstraint[22] = \torqConstraint[23] = 0.05 \si{\newton},
	\MomConstraint[11] = \MomConstraint[21] = 0.085\si{\newton}\si{\meter}\si{\second}, \MomConstraint[12] = \MomConstraint[22]= 0.02\si{\newton}\si{\meter}\si{\second}, \MomConstraint[13] = \MomConstraint[23]= 0.1\si{\newton}\si{\meter}\si{\second}\);
	\begin{itemize}
		\item Initial conditions are set to zero;
		\item Final conditions:  \(\left(\ang{N}{1},\dir{N}{1},\Mom{N}{1},\vel{N}{1},\ang{N}{2},\dir{N}{2},\Mom{N}{2}, \vel{N}{2} \right)= \\ \left({85}{\degree}, (0.25,0.5)^\top\si{\meter}, 0.08  \si{\newton}\si{\meter}\si{\second}, (0.01,0.01)^\top\si{\meter\per\sec}, {85}{\degree}, (0.5,0.25)^\top\si{\meter}, 0.08  \si{\newton}\si{\meter}\si{\second}, (0.01,0.01)^\top\si{\meter\per\sec} \right)\).
	\end{itemize}
\end{itemize}	
	The numerical results for this maneuver are presented in Figures \ref{fig: SE2torq_comparison}-\ref{fig: Traj comparison2}. The multiplexing constraints are clearly satisfied as only one out of two systems is injected with the control action at every time step. Because of the hard boundary constraints the switching frequency is relatively high at the boundary (This is even true for the previous example). 
	
	Since the force vector is zero for large time periods in middle the trajectories are relatively linear in those domains (Figure \ref{fig: Traj comparison2}). 

\begin{figure}[h]
	\centering
	\includegraphics[scale=0.4]{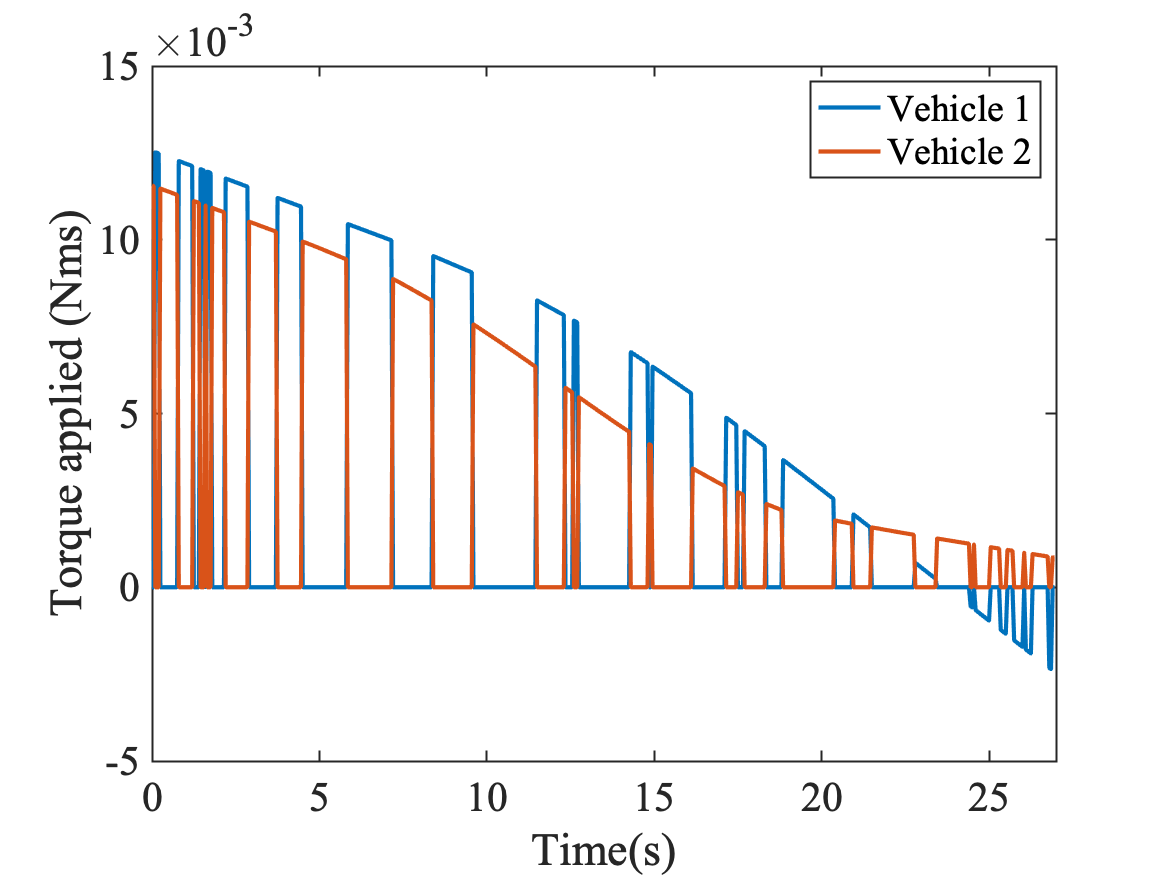}
	\caption{The blue trajectory denotes the torque profile for UW vehicle 1 while the red trajectory denotes the torque profile for UW vehicle 2 under the maneuver \SATmaneuver[\star]  } 
	\label{fig: SE2torq_comparison}
\end{figure}

\begin{figure}[h]
	\centering
	\includegraphics[scale=0.4]{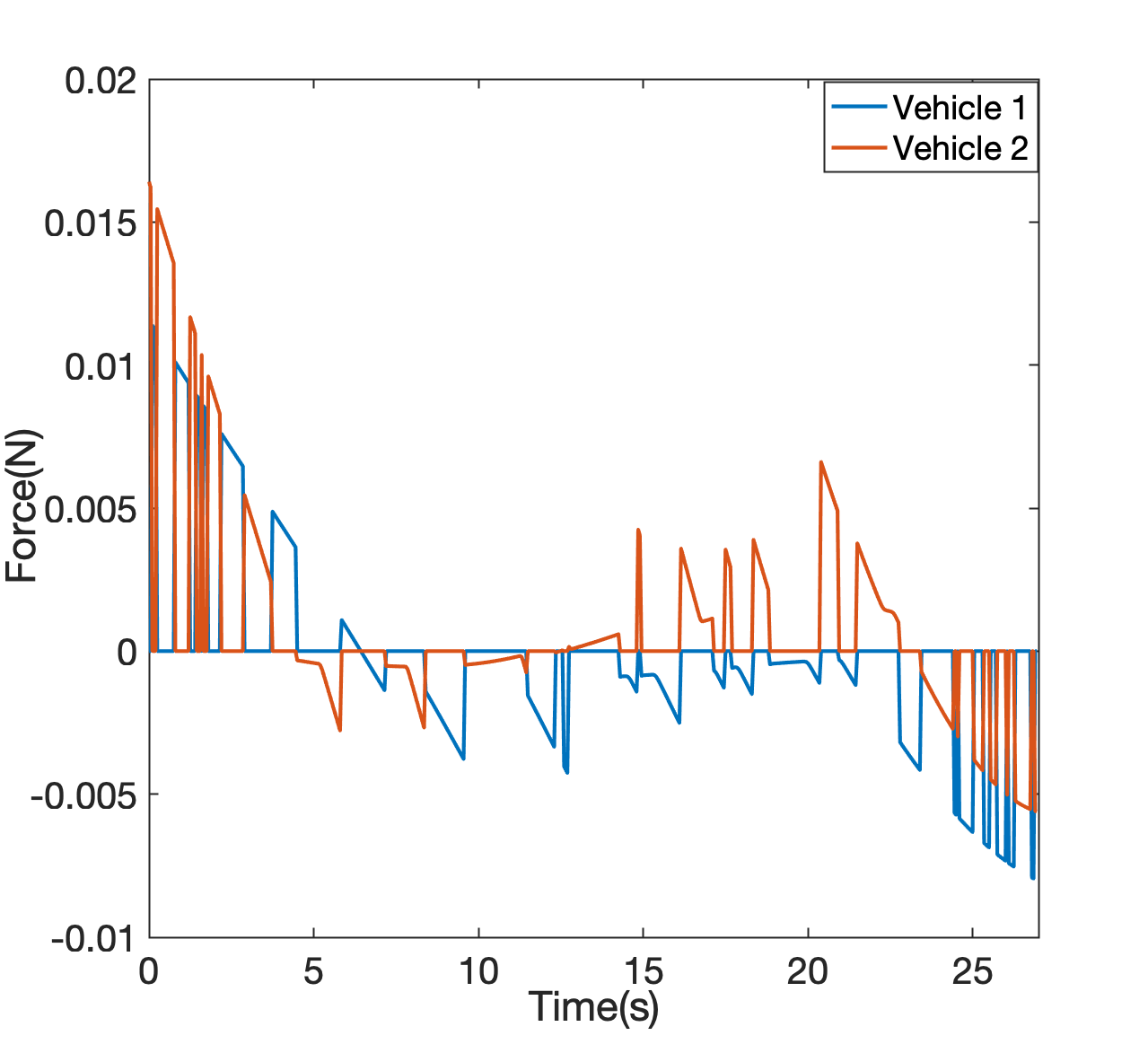}
	\caption{The blue trajectory denotes the force applied in the x-direction on UW vehicle 1 while the red trajectory denotes the same quantity for UW vehicle 2 under the maneuver \SATmaneuver[\star] }
	
	\label{fig: fX comparison}
\end{figure}
\begin{figure}[h]
	\centering
	\includegraphics[scale=0.4]{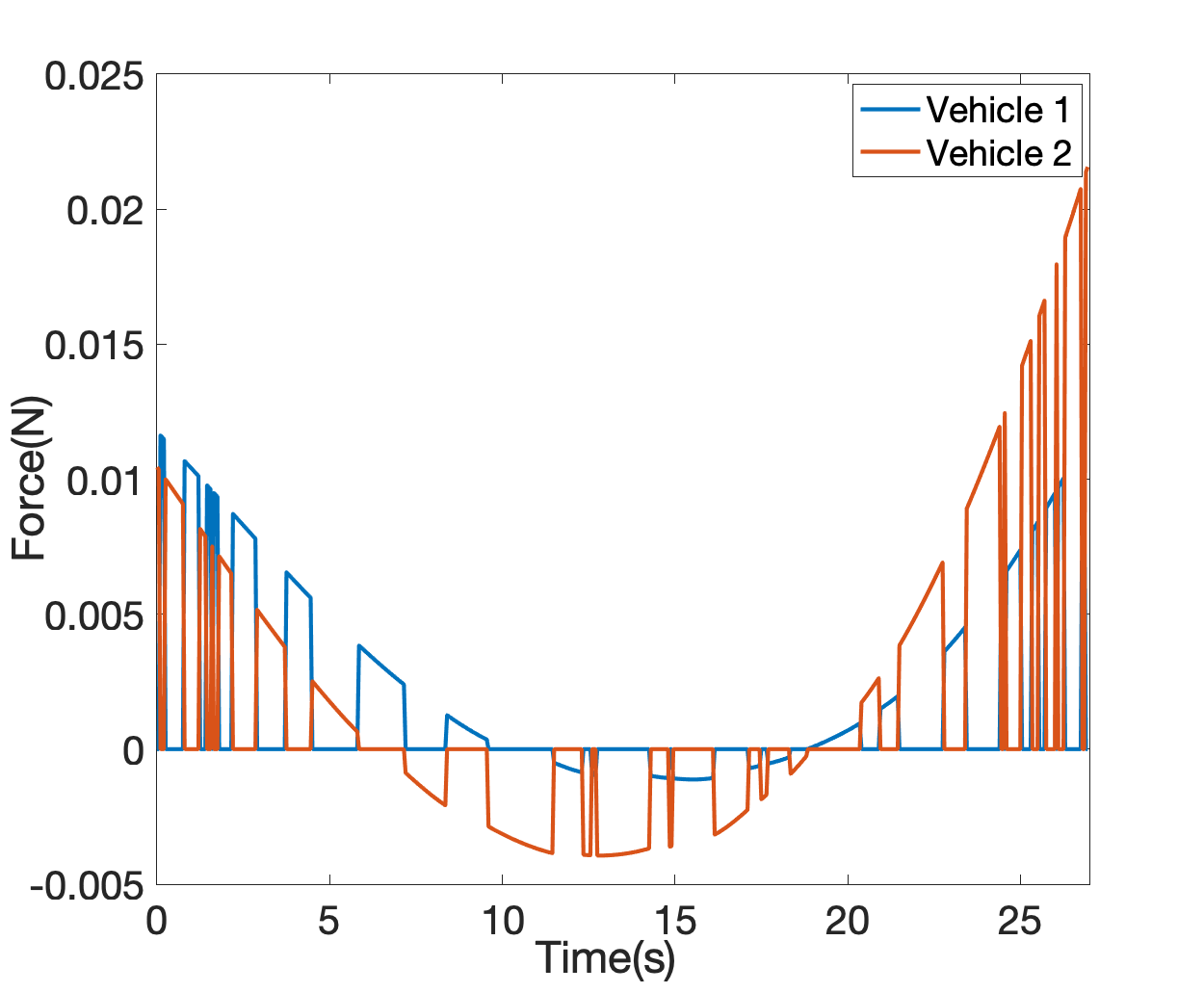}
	\caption{The blue trajectory denotes the force applied in the y-direction on UW vehicle 1 while the red trajectory denotes the same quantity for UW vehicle 2 under the maneuver \SATmaneuver[\star] }
	
	\label{fig: fY}
\end{figure}
\begin{figure}[h]
	\centering
	\includegraphics[scale=0.4]{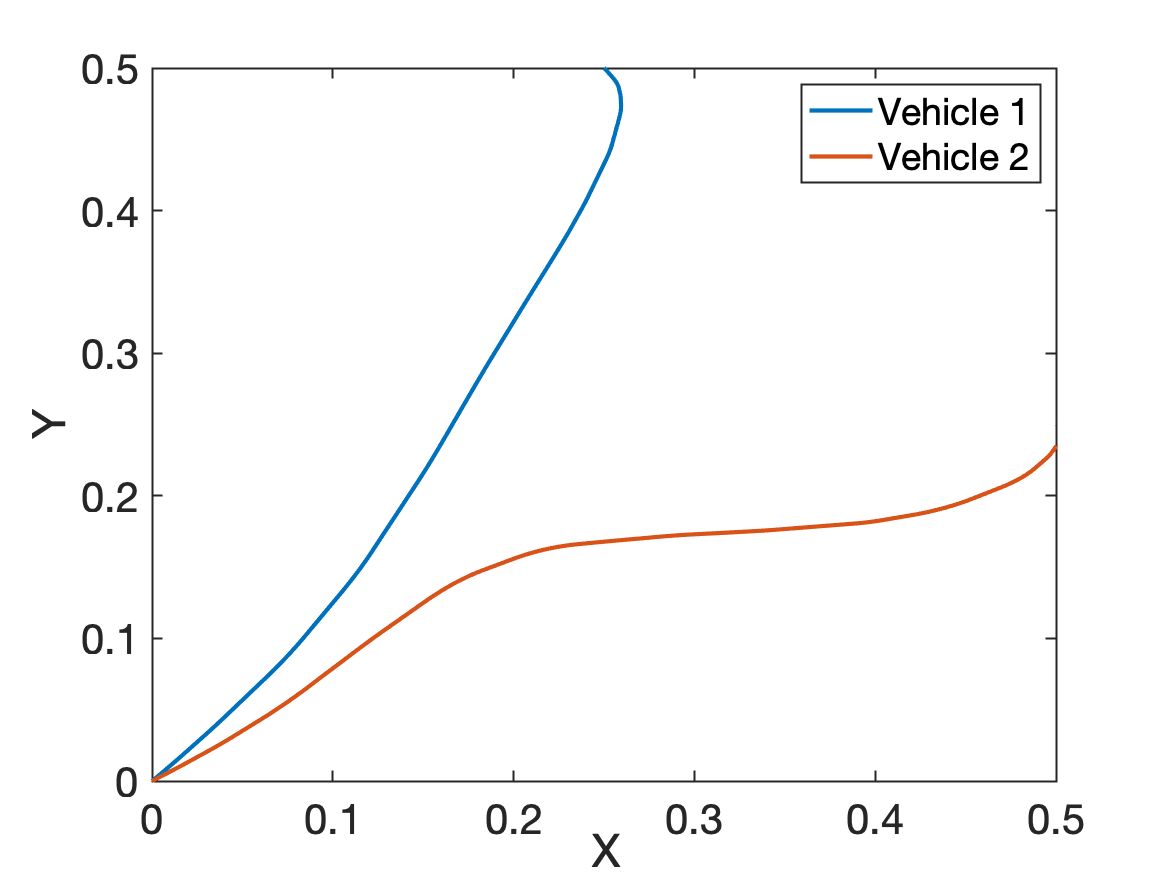}
	\caption{The blue trajectory denotes the path followed by the UW vehicle 1 while the red trajectory depicts the path followed by UW vehicle 2 under the maneuver \SATmaneuver[\star].}
	\label{fig: Traj comparison2}
\end{figure}

	\bibliography{references}
\bibliographystyle{alpha}

\appendix
\section{Discrete-time Pontryagin maximum principle on matrix Lie group}\label{app: DPMP}
 Consider a matrix Lie group \(\liegKP\) and Euclidean space \(\RPower[\dimStateKP]\) for some positive integer \(\dimStateKP \).  Let us define the following discrete-time control system that evolves on the configuration space \(\liegKP\times \RPower[\dimStateKP]\) and is described by 
\begin{equation}\label{eq: KPSystem}
\begin{aligned}
	\begin{cases} \qTKP[t+1] = \qTKP[t]  \jumpKP[t] \left(\qTKP,\xTKP\right)\\
	\xTKP[t+1] = \stateDynKP[t] \left(\qTKP,\xTKP,\uTKP\right)
	\end{cases} \quad \quad \text{for all} \ \timeRange,
\end{aligned}
\end{equation}
where 
\begin{enumerate}[label=(\alph*), leftmargin=*, widest=b, align=left]
	\item \(\qTKP \in \liegKP\), \(\xTKP \in \RPower[\dimStateKP]\) and \(\uTKP \in \controlKP \subsetTo \RPower[\dimControlKP]\) for some positive integer \(\dimControlKP\);
	\item \(\jumpKP: \liegKP \times \RPower[\dimStateKP] \rightarrow \liegKP\) is a smooth map describing the part of the dynamics on the matrix Lie group \(\liegKP \);
	\item \(\stateDynKP[t]: \liegKP \times \RPower[\dimStateKP]\times \RPower[\dimControlKP] \rightarrow \RPower[\dimStateKP] \) is a smooth map describing the part of the dynamics on the Euclidean space \(\RPower[\dimStateKP] \). 
\end{enumerate}

For the preceding control system, we recall a result \cite[Theorem 2.5]{phogat2016discrete} that solves the following optimal control problem: 
\begin{equation}
\label{eq:sopt}
\begin{aligned}
\minimize_{\left(\uTKP\right)_{t=0}^{\nS-1}} &&& \mathscr{J} \left(\qTKP[],\xTKP[],\uTKP[]\right) \defas \sum_{t=0}^{\nS-1} \costKP[t] \left(\qTKP,\xTKP,\uTKP\right) + \costKP[\nS] \left(\qTKP[\nS],\xTKP[\nS]\right)\\
\text{subject to} &&&
\begin{cases}
\text{dynamics \eqref{eq: KPSystem}} & \text{for all} \ \timeRange, \\
\uTKP \in \controlKP \subsetTo \RPower[m] &
\text{for all}\ \timeRange, \\
\constraintKP \left(\qTKP,\xTKP\right) \leq 0 & \text{for all} \ \timeRangeFullSt, \\ 
\left(\qTKP[0],\xTKP[0]\right)=\left(\bar{q}_0,\bar{x}_0\right)
\end{cases} 
\end{aligned} 
\end{equation}
with the following data:
\begin{enumerate}[label=(\alph*), leftmargin=*, widest=b, align=left]
\item \(\qTKP[] \defas \big(\qTKP[0], \qTKP[1], \dots, \qTKP[\nS] \big), \) \(\xTKP[] \defas  (\xTKP[0], \xTKP[1], \dots, \xTKP[\nS])\), \(\uTKP[] \defas (\uTKP[0], \uTKP[1], \dots, \uTKP[\nS]) \);
\item \(\costKP[t] : \liegKP \times \RPower[\dimStateKP] \times \RPower[m] \ra \R\) denotes the cost incurred at each time instant \(\timeRange \);
\item \(\costKP[\nS] : \liegKP \times \RPower[\dimStateKP] \ra \R\) denotes the cost incurred at the final instant \(t = \nS \);
\item \(\constraintKP: \liegKP \times \RPower[\dimStateKP] \ra \RPower[\dimConstraintKP] \) denotes the state constraints that needs to be satisfied at each \(\timeRangeFullSt \), for some positive integer \(\dimConstraintKP \); 
\item \(\left(\bar{q}_0,\bar{x}_0\right)\) denotes the user-defined initial conditions.  
\end{enumerate}
 
 The maps \(\jumpKP, \costKP, \costKP[\nS], \stateDynKP, \constraintKP\), the admissible control action set \(\controlKP\) and the Lie algebra \(\lieaKP\) in the above optimal control problem are required to satisfy Assumption \ref{ass: Existence} to ensure existence of the multipliers that appear in the Theorem \ref{thm:DMP} below:  
\begin{theorem}\cite[Theorem 2.5]{phogat2016discrete}
	\label{thm:DMP}
	Let \(\big(\opuTKP\big)_{t=0}^{\nS-1}\) be an optimal controller that solves the problem \eqref{eq:sopt} with \(\big(\opqTKP, \opxTKP \big)_{t=0}^\nS\) being the corresponding state trajectory. We define the Hamiltonian function, for \(\NU \in \{-1,0\}\), by
	\begin{equation}
	\begin{aligned}\label{main: HamDefinition}
	& \hamdefkp \ni \hamvarkp \mapsto \\
	& \hamkp{\hamvarkp} \defas\NU \costKP[\tau] \left(\qTKP[],\xTKP[],\uTKP[]\right) + \left\langle\zetaKP,\exp^{-1}\left(\jumpKP[\tau] \left(\qTKP[],\xTKP[]\right)\right)\right\rangle_{\lieaKP} + \left\langle \xiKP,\stateDynKP[\tau] \left(\qTKP[],\xTKP[],\uTKP[]\right) \right\rangle  \in \R, 
	\end{aligned}
	\end{equation}
	For \(\timeRangeFullSt\), we define the transformation  
	\begin{align*}
	\lieaKP\dual \ni \zetaKP_t \mapsto \RhoTKP = \Big(\mathcal{D}\ithExpMap[]^{-1}\big((\opqTKP)^{-1}\opqTKP[t+1]\big) \circ \tanLift{e}{\lftAcnKP_{(\opqTKP)^{-1}\opqTKP[t+1]}}\Big)\dual(\zetaKP_t) \in \lieaKP\dual, 
	\end{align*}
	where \(\IdEle[]\) is the identity element of the matrix Lie group \(\liegKP\) and \(\lftAcnKP\) is the left action on the Lie group \(\liegKP\).
	We denote the extremal lift of the state-action trajectory \((\opqTKP,\opxTKP, \opuTKP) \) under the optimal control \(\opuTKP\) at every time instant \(t\) by \(\optraj{t}\), where
	\[\optraj{t}  \defas \left(t,\zetaKP_t,\xiKP_t,\opqTKP,\opxTKP, \opuTKP \right),\]   
 Then there exist an adjoint trajectory \(\big( \zetaKP_t,\xiKP_t\big)_{t=0}^{\nS-1} \subsetTo \lieaKP\dual\times \left(\R^{\dimStateKP}\right)\dual\) , and covectors  \(\big(\muKP\big)_{t=1}^{\nS} \subsetTo \left(\R^{\dimConstraintKP}\right)\dual\), such that the following conditions hold:
	\begin{enumerate}[leftmargin=*, label={\rm (DMP-\roman*)}, widest=iii, align=left]
		\item \label{main:ntriv} non-triviality: the adjoint variables $(\zetaKP_t,\xiKP_t)$ for all \(\timeRange\), the covectors $ \muKP $ for all \(\timeRangeFullSt \), and the scalar ${\nu}$ do not vanish simultaneously;
		
		\item \label{main:dyn} state and adjoint system dynamics:
		\begin{align*}
			\text{states} & \begin{cases}
		\opqTKP[t+1] = \opqTKP[t] \exp\big({\mathcal{D}_{\zetaKP} \hamkp{(\optraj{t} )}}\big)\\
		\opxTKP[t+1] = \mathcal{D}_{\xiKP} \hamkp{(\optraj{t})}
		\end{cases}\\
			\text{adjoints} & \begin{cases}
		\RhoTKP[t-1] = \coAd{\exp\big({-\mathcal{D}_{\zetaKP} \hamkp{(\optraj{t})}}\big)} \RhoTKP + \cotanLift{\IdEle[]}{\lftAcnKP_{\opqTKP}} \Big(\mathcal{D}_{\qTKP[]} \hamkp{(\optraj{t})} + \muKP \mathcal{D}_{\qTKP[]} \constraintKP\left(\opqTKP, \opxTKP\right) \Big)\\
		\xiKP_{t-1} = \mathcal{D}_{\xTKP[]} \hamkp{(\optraj{t})} + \muKP \mathcal{D}_{\xTKP[]} \constraintKP\left(\opqTKP, \opxTKP\right);
		\end{cases}
		\end{align*}
		\item \label{main:trans} transversality:
		\begin{align*}
		\RhoTKP[\nS-1]&= \cotanLift{\IdEle[]}{\lftAcnKP_{\opqTKP[\nS]}}\Big(\nu \mathcal{D}_{\qTKP[]} \costKP[\nS] \left(\opqTKP[\nS], \opxTKP[\nS] \right) + \muKP[\nS] \mathcal{D}_{\qTKP[]} \constraintKP[\nS] \left(\opqTKP[\nS],\opxTKP[\nS] \right)\Big)\\
		\xiKP_{\nS-1}&= \nu \mathcal{D}_{\xTKP[]} \costKP[\nS] \left(\opqTKP[\nS], \opxTKP[\nS] \right) + \muKP[\nS] \mathcal{D}_{\xTKP[]} \constraintKP[\nS] \left(\opqTKP[\nS], \opxTKP[\nS]\right);
		\end{align*}
		
		\item \label{main:hmax} Hamiltonian non-positive gradient: \label{it:hamnpg}
		\[\Big\langle{\mathcal{D}_{\uTKP[]} \hamkp{(\optraj{t})}},{\tilW- \opuTKP}\Big\rangle \leq 0 \quad \text{for all} \ \tilW  \in \suppCone{\controlKP}{\opuTKP},\]
		where \(\suppCone{\controlKP}{\opuTKP}\) is the support cone of \(\controlKP\) with apex at \(\opuTKP\);
		\item \label{main:comp} complementary slackness:
		\begin{align*}
		\muKP \HadPr \constraintKP[t](\opqTKP,\opxTKP) = 0 \in \RPower[\dimConstraintKP] \quad \text{for all} \ \timeRangeFullSt; 
		\end{align*}
		
		\item \label{main:npos} non-positivity:
		\[ \vecCompare{\muKP}{0} \quad \text{for all}\ \timeRangeFullSt\]
	\end{enumerate}
\end{theorem}

\section{Properties of the direct product of matrix Lie groups}\label{app: DiPr}
Here we provide several important properties of the direct product of Lie groups which are utilized in \S\ref{sec: Proof}.

\begin{definition}[{\cite[p.\ 259]{bullo2004geometric}}]
	\label{ProjMap}
	If \(\grpG\) and \(\grpH\) are two groups, then a \textit{group homomorphism} is a map \(\grpG \ni \grp \mapsto \homo(\grp) \in \grpH\) that satisfies \(\homo(\grp_1 \grp_2) = \homo(\grp_1) \homo(\grp_2)\) for all \(\grp_1, \grp_2 \in \grpG\). A \textit{Lie group homomorphism} is a \emph{smooth} group homomorphism between Lie groups. 
\end{definition}

By Definition \ref{ProjMap} the projection map \kappaMap defined in \S \ref{sec: Main Result} is a Lie group homomorphism from \(\gTot\) to \(\grpMfdIndex\) for each \(\iRange\); a fact that follows from the direct product group structure on \(\gTot\). 

\begin{lemma}[{\cite[Proposition 5.3.6]{ref:RudSch-13}}]
	\label{lem: HomomorphismLemma}
	Consider two Lie groups \(\grpG\) and \(\grpH\) with identity elements \(\ithId[\grpG]\) and \(\ithId[\grpH]\), respectively, and let \(\lieAlg\) be the Lie algebra of \(\grpG\). For a Lie group homomorphism \( \lgHomo: \grpG \ra \grpH\) and for any \(\XI \in \lieAlg\) we have
	\begin{align}
	\lgHomo \circ \ithExpMap[\grpG](\XI)  = \ithExpMap[\grpH]  \circ\ \tanLift{\ithId[\grpG]}{\lgHomo(\XI)},
	\end{align}
	where \(\tanLift{\ithId[\grpG]}{\lgHomo}\) is the tangent map of \(\lgHomo\) at \(\ithId[\grpG]\) and \(\ithExpMap[\grpG]\) and \(\ithExpMap[\grpH]\) are the exponential maps of the Lie groups \(\grpG\) and \(\grpH\), respectively. 
\end{lemma}

\begin{lemma}[{\cite[Example 5.3.16]{ref:RudSch-13}}]
	\label{lem: expdist}
	Let \(\grpGNum[1]\) and \(\grpGNum[2]\) be two Lie groups with Lie algebras \(\lieAlgNum[1]\) and \(\lieAlgNum[2]\), respectively. Consider the Lie group \(\grpG \defas \grpGNum[1] \times \grpGNum[2]\) equipped with the direct product group structure and let \(\grpG \ni \grp \mapsto \projmap(\grp) \in \grpGNum\) denote the projection for \(i \in \{1,2\}\). Under the natural identification of the Lie algebra \(\lieAlg\) of \(\grpG\) with \(\lieAlgNum[1] \DirSum \lieAlgNum[2]\), we have
	\begin{align*}
	\ithExpMap[\grpG](\vf[1], \vf[2]) = \big(\ithExpMap[{\grpGNum[1]}](\vf[1]), \ithExpMap[{\grpGNum[2]}](\vf[2])\big)\quad\text{for }\vf[i] \in \lieAlgNum, i \in \{1,2\}. 
	\end{align*}
\end{lemma}
On the basis of previous lemma, we can establish that the inverse of the exponential map also splits into factors:

\begin{lemma}\label{lem: invexpdist}
	Let \(\grpGNum[1]\) and \(\grpGNum[2]\) be two Lie groups with Lie algebras \(\lieAlgNum[1]\) and \(\lieAlgNum[2]\), respectively. Consider the Lie group \(\grpG \defas \grpGNum[1] \times \grpGNum[2]\) equipped with the direct product group structure and let \(\grpG \ni \grp \mapsto \projmap(\grp) \in \grpGNum\) denote the projection for \(i \in \{1,2\}\). Under the natural identification of the Lie algebra \(\lieAlg\) of \(\grpG\) with \(\lieAlgNum[1] \DirSum \lieAlgNum[2]\), we have
	\begin{align*}
	\ithExpMap[\grpG]^{-1}(\grp_1, \grp_2) = \big(\ithExpMap[{\grpGNum[1]}]^{-1}(\grp_1), \ithExpMap[{\grpGNum[2]}]^{-1}(\grp_2)\big)\quad\text{for }\grp_i \in \grpGNum[i], i \in \{1,2\}.
	\end{align*}
	
\end{lemma}
\begin{proof}
	As a consequence of the direct product group structure, for any \(\vect[] \in \lieAlg \) we can find some \(\vect[1] \in \lieAlgNum[1] \) and \(\vect[2] \in \lieAlgNum[2] \) such that \(v = (\vect[1], \vect[2]) \); moreover, this representation is unique \cite[Chapter 5]{ref:RudSch-13}. Therefore, we can split the inverse exponential map as follows
	\begin{align*}
	\ithExpMap[\grpG]^{-1}(\grp_1, \grp_2) = (\vect[1], \vect[2]),
	\end{align*}
	assuming \((\grp_1,\grp_2) \) lies in a set where the inverse of the exponential map is well defined (refer \ref{asm:2}). The preceding equation implies that
	\begin{align*}
	(\grp_1, \grp_2) =  \ithExpMap[\grpG](\vect[1], \vect[2]) = \big(\ithExpMap[{\grpGNum[1]}](\vect[1]), \ithExpMap[{\grpGNum[2]}](\vect[2])\big).
	\end{align*}
	It follows that \(\vect[1] = \ithExpMap[{\grpGNum[1]}]^{-1}(\grp_1) \) and \(\vect[2] = \ithExpMap[{\grpGNum[2]}]^{-1}(\grp_2) \). 
\end{proof}

Similarly, the tangent maps also exhibit analogous splitting as is proved in the following lemma.
\begin{lemma}\label{lem: TanMapDist}
	Let \(\grpGNum[1], \grpGNum[2]\) be two Lie groups with Lie algebras \(\lieAlgNum[1], \lieAlgNum[2]\) respectively. Consider the Lie group \(\grpG \defas \grpGNum[1] \times \grpGNum[2]\) equipped with the direct product group structure and let \(\grpG \ni \grp \mapsto \projmap(\grp) \in \grpGNum\) denote the natural projection operation for \(i \in \{1,2\}\). The identity element of the Lie group \(\grpG\) is denoted by \(\IdEle[] \) and that of the Lie group \(\grpGNum[i]\) by \(\IdEle[i]\) \((\) for every \(i \in \{1,2\} )\). Also, the left action on the Lie group \(\grpG\) is denoted by \(\lftAcn[]\) and on the Lie group \(\grpGNum[i]\) by \(\lftAcn[(i)]\) \((\)where \(i \in \{1,2\} )\). Under the natural identification of the Lie algebra of \(\grpG\) $($denoted by \(\lieAlg\)$)$ with \(\lieAlgNum[1] \DirSum \lieAlgNum[2]\), we have
	\begin{align*}
	\tanLift{{\IdEle[]}}{{\lftAcn[]_{(\grp_1,\grp_2)}}}\big( \vect[1], \vect[2] \big) = \Big(\tanLift{{\IdEle[1]}}{{\lftAcn[(1)]_{(\grp_1)}}}\big( \vect[1]\big),\tanLift{{\IdEle[2]}}{{\lftAcn[(2)]_{(\grp_2)}}}\big( \vect[2]\big) \Big),
	\end{align*}
	for \(\grp_i \in \grpGNum[i]\), \(\vect[i] \in \lieAlgNum[i] \) for \(i\in \{1,2\} \). 
\end{lemma}

\begin{proof}
	For some \(\epsilon > 0 \), consider a smooth curve \(]-\epsilon,\epsilon[ \ni t \mapsto \curv(t) \in \grpG \) such that \(\curv(0) = \IdEle[]\) and \(\left.\frac{d}{dt}\right|_{t=0}\curv(t) = \vect[] \in \tanLift{{\IdEle[]}}{\grpG} \). As a consequence of the direct product structure on the Lie group \(\grpG \), we have a unique decomposition of \(\curv(t)\) into \( \big(\curv_1(t), \curv_2(t) \big) \) where \(\curv_i(\cdot) \in \grpGNum[i]\), \(\curv_i(0) = \IdEle[i] \) and \(\left. \frac{d}{dt} \right|_{t=0}\curv_i(t) = \vect[i] \in \tanLift{{\IdEle[]}}{\grpGNum[i]} \) for \(i \in \{1,2\}\) such that \(\vect[] = (\vect[1], \vect[2]) \). We apply the left action \(\lftAcn[]_{\grp}(\cdot) \) to the trajectory \(\curv(\cdot) \) and compute the derivative at \(t = 0 \) we get  
	\begin{align*}
	\tanLift{{\IdEle[]}}{{\lftAcn[]_{(\grp_1,\grp_2)}}}\big( \vect[1], \vect[2] \big) &= \frac{d}{dt}\Big|_{t=0} \lftAcn[]_{(\grp_1,\grp_2)} \big(\curv(t)\big) \\
	&= \frac{d}{dt}\Big|_{t=0} \Big(\lftAcn[(1)]_{\grp_1}\big(\curv_1(t)\big),\lftAcn[(2)]_{\grp_2}\big(\curv_2(t)\big) \Big) \\
	&= \Big(\tanLift{{\IdEle[1]}}{{\lftAcn[(1)]_{(\grp_1)}}}\big( \vect[1]\big),\tanLift{{\IdEle[2]}}{{\lftAcn[(2)]_{(\grp_2)}}}\big( \vect[2]\big) \Big).
	\end{align*}
\end{proof}

\begin{lemma}{\cite[Example 5.4.8]{ref:RudSch-13}}\label{lem: adjDist}
	Let \(\grpGNum[1], \grpGNum[2]\) be two Lie groups with Lie algebras \(\lieAlgNum[1], \lieAlgNum[2]\) respectively and let \(\grpG \defas \grpGNum[1] \times \grpGNum[2]\) has \(\lieAlg\) as its Lie algebra. Under the natural identification of \(\lieAlg\) with \(\lieAlgNum[1] \DirSum \lieAlgNum[2]\), we have 
	\begin{align*}
	\ad{(\grp_1,\grp_2)}{(\vf[1], \vf[2])} = (\ad{\grp_1}{\vf[1]}, \ad{\grp_2}{\vf[2]}), 
	\end{align*} 
	where \(\grp_i \in \grpGNum[i]\) and \(\vf \in \lieAlgNum[i]\) for \(i \in \{1,2\} \).
\end{lemma}
\begin{lemma}\label{Prop: coadjoint}
	Let \(\grpGNum[1], \grpGNum[2]\) be two Lie groups with Lie algebras \(\lieAlgNum[1], \lieAlgNum[2]\) respectively and let \(\grpG = \grpGNum[1] \times \grpGNum[2]\) has \(\lieAlg\) as its Lie algebra. Under the natural identification of \(\lieAlg\) with \(\lieAlgNum[1] \DirSum \lieAlgNum[2]\), we have  
	\begin{align*}
	\coAd{(\grp_1, \grp_2)}{(\dAlg[1], \dAlg[2])}  =  \big(\coAd{\grp_1}(\dAlg[1]), \coAd{\grp_2}(\dAlg[2])\big),
	\end{align*}
	where \((\dAlg[1], \dAlg[2]) \in \dlieAlgNum[1] \DirSum \dlieAlgNum[2]\), \(\grp_i \in \grpGNum[i]\) for \(i \in \{1,2\} \).
\end{lemma}

\begin{proof}
	Using Lemma \ref{lem: adjDist} and by definition of co-adjoint action (in Definition \ref{def: AdjointDef}) we have
	\begin{align*}
	\Big \langle \coAd{(\grp_1, \grp_2)} (\dAlg[1], \dAlg[2]) , (\vf[1], \vf[2])\Big\rangle &= \Big \langle (\dAlg[1], \dAlg[2]), \ad{(\grp_1, \grp_2)}(\vf[1], \vf[2]) \Big\rangle \\
	&= \Big \langle (\dAlg[1], \dAlg[2]), (\ad{\grp_1}(\vf[1]), \ad{\grp_2}(\vf[2])) \Big \rangle \\
	&= \Big \langle \big(\ad{\grp_1}^{*}(\dAlg[1]), \ad{\grp_2}^{*}(\dAlg[2])\big), (\vf[1], \vf[2]) \Big \rangle,
	\end{align*}
	for all \((\vf[1], \vf[2]) \in \lieAlgNum[1] \oplus \lieAlgNum[2]\).
\end{proof}

\section{Properties of the map \(z\)}\label{app: app_z}
In this part, we present some of the useful observations about the mapping \(z\) defined in \eqref{eq: constraintEq}. 
\begin{lemma}\label{rem: IFFRemark}
	For the mapping \(z \) defined in \eqref{eq: constraintEq}, we have \(\auxDyn = (0,0)^\top \)  if and only if \(U \in \uhash\), where \(\uhash\) is defined in \eqref{eq: AdmissibleActionSet}. 
\end{lemma}
\begin{proof}
	For any \(U \in \uhash \), by definition of \(\uhash \) we have \(U = \big(\uSig[i]\big)_{i=1}^{\numPlants}\) such that atmost one of \(\uSig[i] \), for \(\iRange \), will be non-zero. This implies that \(\auxDyn = (0,0)^\top\). 
	Now we prove that if \(\auxDyn[V] = (0,0)^\top\), for some \(V \in \ustar \), then \(V\) belongs to \(\uhash \subset \ustar\). If  \(\auxDyn[V] = (0,0)^\top\), then 
	\begin{align}
		\begin{pmatrix}
		0 \\ 0
		\end{pmatrix} = \auxDyn[V] = \sum_{i = 1}^{\numPlants-1} \sum_{j = i+1}^{\numPlants} \bigg( \bigl\|\vSig[i]\bigr\|^2 \bigl\|\vSig[j]\bigr\|^2 \begin{pmatrix} 1 \\ 1 \end{pmatrix} + \vSig[i] \opr \vSig[j] \begin{pmatrix} 1 \\ -1 \end{pmatrix}   \bigg) 
	\end{align}
	where \(V = \big(\vSig\big)_{i=1}^{\numPlants} \) such that \(\vSig[i] \in \ithAdmControl \) for \(\iRange \). As a consequence of preceding equation, we get \(\|\vSig[i]\|^2\|\vSig[j]\|^2 = 0 \) for all \(i, j \in [\numPlants]\dual \) and \(i \neq j \), which indicates that \(V \) lies on one of the branches of ``star''-shaped admissible control action set \(\uhash \). 
\end{proof}
\begin{lemma}\label{lem: zProp}
	For some positive integer \(R\), consider a set of vectors \(\sU[1], \sU[2],\dots, \sU[R] \in \ustar \) such that \(\sum_{k=1}^{R}\auxDyn[{\sU[k]}] = (0, 0)^\top \). Then \(\auxDyn[{\sU[k]}] = (0,0)^\top \) for all \(k \in [R]\dual\).
\end{lemma}
\begin{proof}
	Using the definition of the mapping \(z \) as given in \eqref{eq: constraintEq}, we get 
	\begin{align*}
		\begin{pmatrix}
			0 \\ 0
		\end{pmatrix} = \sum_{k=1}^{R}\auxDyn[{\sU[k]}] =\sum_{k=1}^{R} \sum_{i = 1}^{\numPlants-1} \sum_{j = i+1}^{\numPlants} \bigg( \bigl\|\uSig[i]_k\bigr\|^2 \bigl\|\uSig[j]_k\bigr\|^2 \begin{pmatrix} 1 \\ 1 \end{pmatrix} + \uSig[i]_k \opr \uSig[j]_k \begin{pmatrix} 1 \\ -1 \end{pmatrix}   \bigg), 
	\end{align*}
	where \(\uSig[i]_k \in \ithAdmControl[i] \) for all \(k \in [R]\dual \), \(\iRange\). The preceding vector equation gives us the following set of equations 
	\begin{equation}\label{eq: Consequence}
	\begin{aligned}
		\sum_{k=1}^{R} \sum_{i = 1}^{\numPlants-1} \sum_{j = i+1}^{\numPlants}  \bigl\|\uSig[i]_k\bigr\|^2 \bigl\|\uSig[j]_k\bigr\|^2 &= \ \ \ \sum_{k=1}^{R} \sum_{i = 1}^{\numPlants-1} \sum_{j = i+1}^{\numPlants} \uSig[i]_k \opr \uSig[j]_k, \\\sum_{k=1}^{R} \sum_{i = 1}^{\numPlants-1} \sum_{j = i+1}^{\numPlants}  \bigl\|\uSig[i]_k\bigr\|^2 \bigl\|\uSig[j]_k\bigr\|^2 &= -\sum_{k=1}^{R} \sum_{i = 1}^{\numPlants-1} \sum_{j = i+1}^{\numPlants} \uSig[i]_k \opr \uSig[j]_k .
	\end{aligned}
	\end{equation}
	In view of \eqref{eq: Consequence}, we have 
	\begin{align}\label{eq: zFirst}
		 \bigl\|\uSig[i]_k\bigr\|^2 \bigl\|\uSig[j]_k\bigr\|^2 = 0 \quad \text{for all}\ k \in [R]\dual, i, j \in [\numPlants]\dual \ \text{and} \ i \neq j, 
	\end{align}
	which by the definition of the product \(\opr \) renders 
	\begin{align}\label{eq: zSec}
		\uSig[i]_k \opr \uSig[j]_k = 0 \quad \text{for all}\ k \in [R]\dual, i, j \in [\numPlants]\dual \ \text{and} \ i \neq j. 
	\end{align}
	As a consequence of \eqref{eq: zFirst} and \eqref{eq: zSec}, we have \(\auxDyn[{\sU[k]}] = (0, 0)^\top \) for all \(k \in [R]\dual\). 
\end{proof}

\end{document}